\documentclass[12pt]{article}
\usepackage[utf8]{inputenc}
\usepackage[a4paper]{geometry}
\usepackage{color}
\usepackage{float} 
\usepackage{array}
\usepackage{verbatim}
\usepackage{multirow}
\usepackage{amsmath}
\usepackage{amsthm}
\usepackage{graphicx}
\usepackage{setspace}
\usepackage[authoryear]{natbib}
\usepackage[unicode=true,
 bookmarks=false,
 breaklinks=false,pdfborder={0 0 0},pdfborderstyle={},backref=false,colorlinks=true]
 {hyperref}
\hypersetup{
 citecolor=blue, linkcolor=blue, urlcolor=blue}

\makeatletter

\providecommand{\tabularnewline}{\\}

\numberwithin{equation}{section}
\numberwithin{figure}{section}
\theoremstyle{definition}
\newtheorem{defn}{\protect\definitionname}
\theoremstyle{plain}

\theoremstyle{plain}
\newtheorem{prop}{\protect\propositionname}
\theoremstyle{plain}
\newtheorem{thm}{\protect\theoremname}
\theoremstyle{remark}
\newtheorem{rem}{\protect\remarkname}
\theoremstyle{plain}
\newtheorem{lem}{\protect\lemmaname}
\theoremstyle{plain}
\newtheorem{cor}{\protect\corollaryname}
\newtheorem{fact}{Fact}

\usepackage[T1]{fontenc}
\usepackage{lmodern}
\definecolor{green}{RGB}{00, 180, 00}
\definecolor{red}{RGB}{180, 00, 00}
\usepackage{amsfonts}

\makeatother

\providecommand{\corollaryname}{Corollary}
\providecommand{\definitionname}{Definition}
\providecommand{\lemmaname}{Lemma}
\providecommand{\propositionname}{Proposition}
\providecommand{\remarkname}{Remark}
\providecommand{\theoremname}{Theorem}
\def\t{\theta}
\def\D{\Delta}
\def\e{\epsilon}

\begin{document}
\title{{\large{}Sampling Dynamics and Stable Mixing in Hawk--Dove Games\\}}
\author{Srinivas Arigapudi\thanks{
Faculty of Industrial Engineering and Management, Technion. \protect\href{mailto:arigapudi@campus.technion.ac.il}{arigapudi@campus.technion.ac.il}.}
\and Yuval Heller\thanks{Department of Economics, Bar-Ilan University.  \protect\href{mailto:yuval.heller@biu.ac.il}{yuval.heller@biu.ac.il}. }
\and Amnon Schreiber\thanks{Department of Economics, Bar-Ilan University.  \protect\href{mailto:amnon.schreiber@biu.ac.il}{amnon.schreiber@biu.ac.il}. }
\thanks{We thank Daniel Friedman, Itay Kavaler, Ron Peretz, Ernst Schulte-Geers, Daniel Stephenson, Jiabin Wu, and seminar audiences at Bar-Ilan University and the University of Haifa for various helpful comments. We thank  Luis R. Izquierdo and Segismundo S. Izquierdo for their help in implementing the simulations using ABED software. YH and SA  gratefully acknowledge the financial support of the European Research
Council (\#677057) and the Israeli Science Foundation (\#2443/19). SA is supported in part at the Technion by a Fine Fellowship.
}}
\maketitle
\begin{abstract}
The hawk--dove game admits two 
types of equilibria: an asymmetric pure equilibrium in which players in one population play “hawk'' and players in the other population play ``dove,'' and an 
inefficient symmetric  mixed equilibrium, 
in which hawks are frequently matched against each other.
The existing  
literature shows that populations will converge to playing one of the 
pure equilibria  from almost any initial state. By contrast, we show that plausible sampling dynamics, in which agents occasionally revise their actions by observing either opponents' behavior or payoffs in a few past interactions, can induce the opposite result:  
convergence to one of the inefficient mixed stationary states.

\end{abstract}

\section{Introduction}
The hawk--dove game is 
widely used to study situations of conflict between strategic participants.\footnote{
A few examples of the applications of the hawk--dove game 
(also known as the chicken game; see, e.g., \citealp{rapoport1966game,aumann1987correlated}) are: provision of public goods (\citealp{lipnowski1983voluntary}), nuclear deterrence between superpowers (\citealp{brams1987threat};  \citealp{dixit2019we}), industrial disputes (\citealp{bornstein1997cooperation}), bargaining problems (\citealp{brams2001fallback}), conflicts between countries over contested territories (\citealp{baliga2012strategy,baliga2020strategy}), and task allocation among members of a team (\citealp{herold2020evolution}).}
As a simple motivating example,
consider a situation in which a buyer (Player 1) and a seller (Player 2) bargain over the price of an asset (e.g., a house). Each player has two possible bargaining strategies (actions): insisting on a more favorable price (referred to as being a ``hawk''), or agreeing to a less favorable price in order to close the deal (being a ``dove''). The payoffs of the game are presented in Table \ref{tab:canonical-Hawk--dove-Games-1}. 
Two doves agree on a price that is equally favorable to both sides, and obtain a relatively high payoff, which is normalized to 1. A hawk obtains a favorable price when being matched with a dove, which yields her an additional gain of $g\in(0,1)$, at the expense of her dovish opponent.\footnote{
Our formal model studies a broader class of generalized hawk--dove games, in which the gain of a hawkish player might differ from the loss of the dovish opponent (as shown in Table \ref{tab:Generalized-Hawk--dove-Games} in Section \ref{sec:model}).} 
 Finally, two hawks obtain the lowest payoff of 0, due to a substantial probability of bargaining failure.\footnote{
Our one-parameter payoff matrix is equivalent to the commonly used two-parameter payoff matrix (\citealp{smith1982evolution}), in which a dove obtains $\frac{V}{2}$ against another dove and 0 against a hawk, and a hawk obtains $\frac{V-C}{2}$ against another hawk and $V$ against a dove. Specifically, our one-parameter matrix is obtained from the two-parameter matrix by the affine transformation of adding the constant $\frac{C-V}{2}$ and dividing all payoffs by $\frac{C}{2},$ followed by substituting $g\equiv\frac{V}{C}.$}
 Observe that large values of $g$ (close to 1) correspond to environments that are advantageous to hawks (i.e., being a hawk yields a higher expected payoff against an opponent who might play either action with equal probability), small values of $g$ correspond to environments that are advantageous to doves, and values of $g$ that are close to 0.5 correspond to approximately balanced environments.

\begin{table}
\begin{centering}
\caption{Payoff Matrix of the Standard Hawk--Dove Game ($g\in\left(0,1\right)$) \label{tab:canonical-Hawk--dove-Games-1}}
\par\end{centering}
\smallskip{}
\centering
    \setlength{\extrarowheight}{2pt}
    \begin{tabular}{cc|c|c|}
      & \multicolumn{1}{c}{} & \multicolumn{2}{c}{\textcolor{red}{Player 2}}\\
      & \multicolumn{1}{c}{} & \multicolumn{1}{c}{$\textcolor{red}{h}$}  & \multicolumn{1}{c}{$\textcolor{red}{d}$} \\\cline{3-4}
      \multirow{2}*{\textcolor{blue}{Player 1}}  & $\textcolor{blue}{h}$ & $\textcolor{blue}{0},\textcolor{red}{0}$ & $\textcolor{blue}{1+g},\textcolor{red}{1-g}$ \\\cline{3-4}
      & $\textcolor{blue}{d}$ & $\textcolor{blue}{1-g},\textcolor{red}{1+g}$ & $\textcolor{blue}{1},\textcolor{red}{1}$ \\\cline{3-4}
    \end{tabular}
  \end{table}

The hawk–dove game admits three Nash equilibria: two asymmetric pure equilibria, and 
an inefficient symmetric mixed equilibrium. In the pure equilibria (in which one of the players plays hawk while the opponent plays dove), all conflicts are avoided at the cost of inequality, as the payoff of the hawkish player is substantially higher than that of the dovish opponent. By contrast, in the symmetric mixed equilibrium 
both players obtain the same expected payoff,
yet this payoff is relatively low due to the positive probability of a conflict arising between two hawks. 

A natural question is to ask which equilibrium is more likely to obtain. Standard game theory is not helpful in answering this question, as all these Nash equilibria satisfy all the standard refinements (e.g., perfection). By contrast, the dynamic (evolutionary) approach can yield sharp predictions (for textbook expositions, see \citealp{weibull1997evolutionary,sandholm2010population}).

\paragraph{Revision dynamics}

Consider a setup in which pairs of agents from two infinite populations are repeatedly matched at random times (each such match of an agent from population 1 is against a new opponent from population 2).\footnote{Our paper focuses on two-population dynamics in which players condition their play on their role in the game. By contrast, in a one-population model, players cannot condition their play on their role in the game. The predictions of the one-population model are briefly discussed in Remark \ref{rem-two-one}.} Agents occasionally die (or, alternatively, agents occasionally receive opportunities to revise their actions). New agents observe some information about the aggregate behavior and the payoffs, and use this information to choose the action they will play in all future encounters. We are interested in characterizing the stable rest points of such revision dynamics, which can be used as an equilibrium refinement.

Most existing models assume that the revision dynamics are \emph{monotone} (also known as sign-preserving) with respect to the payoffs: the frequency of the strategy that yields the higher payoff (among the two feasible strategies) increases. A key result in evolutionary game theory is that in a hawk--dove game, all monotone
(two-population) revision dynamics converge to the asymmetric pure equilibria  from almost any initial state (henceforth, \emph{global convergence}; see \citealp{smith1976logic}, for the classic analysis, \citealp{smith1982evolution}, for the textbook presentation, \citealp{sugden1989spontaneous}, for the economic implications, and \citealp{oprea2011separating}, for the general dynamic result.)  Thus, the existing literature predicts that an efficient convention 
will emerge in which trade always occurs and most of the surplus goes to one side of the market. Casual observation suggests that this prediction might not fit well the behavior in situations such as the motivating example, in which the surplus of trade is typically divided relatively equally between the two sides of the market, and in which bargaining frequently fails.

In many applications, precise information about the aggregate behavior in the population  may be difficult or costly to obtain. In such situations, new agents have to infer the aggregate behavior in the population from a small sample of other players. 
In what follows, we 
study two plausible inference procedures, both of which violate monotonicity. The first procedure is the \emph{action-sampling dynamics} (also known as sampling best-response dynamics;  \citealp{sandholm2001almost,osborne2003sampling}). In these dynamics, each new agent observes the behavior of $k$ random opponents, and then adopts the action that is a best reply to her sample (with an arbitrary tie-breaking rule). 

In some applications, new agents may not be able to observe  opponents' actions, or they may lack information about the payoff matrix. Plausible revision dynamics in such situations are the \emph{payoff-sampling dynamics}  (also known as best experienced payoff dynamics; \citealp{osborne1998games,sethi2000stability}). 
In these dynamics, each new agent observes the payoffs obtained by incumbents of her own population in $k$ interactions in which these incumbents played hawk, and in $k$ interactions in which these incumbents played dove. Following these observations, the new agent adopts the action that yielded the higher mean payoff.

We analyze both sampling dynamics in the hawk--dove game. In our analysis, we allow agents to have heterogeneous sample sizes (i.e., each new agent is endowed with a sample size of $k$ that is randomly chosen from an exogenous distribution). 
It is relatively straightforward to show that these dynamics admit at least three stationary states: 
two asymmetric pure states, and an inefficient symmetric state.\footnote{
Although the  symmetric stationary state does not coincide with the symmetric Nash equilibrium, they share similar qualitative properties: namely, symmetry between the two populations, and inefficiency induced by frequent matching of two hawks.}$^{,}$
\footnote{In some cases, the dynamics admit also asymmetric mixed states,
as is demonstrated in Section \ref{sec:Numeric-Analysis}.} 

\paragraph{Global convergence to mixed states} 
We show that sampling dynamics can yield qualitatively different results compared to monotone dynamics in the hawk--dove game. Specifically, our first main result (Theorem \ref{thm:global-mixed} and Corollary \ref{cor:g=l}) presents a simple condition for sampling dynamics to induce the opposite result (relative to monotone dynamics), namely, convergence to one of the mixed stationary states from any interior state. Roughly speaking, this happens iff the following three conditions hold:
 (1) sufficiently many agents have small sample sizes, (2) 
the population includes some agents with relatively large sample sizes (though still below a threshold that is increasing in $|g-0.5|$), and (3) $g$  is not too close to 0.5. 

The key to this result is to characterize when populations that start near one of the pure equilibria get away from this equilibrium.
Assume that initially almost all ($1-\e$) buyers are hawks, and almost all ($1-\e$) sellers are doves. Consider a new agent who bases her decision  on a sample of $k$ actions. The sample will have a single occurrence of the rare action with a probability of about $k\cdot \epsilon$ (while the probability of having two or more occurrences of the rare action is negligible). The presence of agents with small samples is necessary for moving away from the pure equilibrium, because for dove-favorable games with $g<0.5$ a single occurrence of a rare action can induce new sellers to be hawks only if their samples are small 
(similarly, only new buyers  with small samples  change their behavior after a single occurrence of a rare action in hawk-favorable games with $g>0.5$).

In dove-favorable games, a single occurrence of a rare action can induce a new buyer to be a dove  even for 
relatively large sample sizes up to a threshold, 
where this threshold is larger the farther $g$ is from 0.5 
(and the same holds for inducing new sellers to be hawks in hawk-favorable games). Having a larger sample size $k$ (up to the above-mentioned threshold) helps to increase the frequency in which rare actions are observed,  which explains why having a sufficiently large expected sample size (Condition (2)) is necessary for moving away from the pure equilibrium. Finally, when the game is balanced between the two actions (i.e., $g$ is close to 0.5), a single occurrence of a rare action can change the behavior of a new agent (in both populations) only if her sample is relatively small, but in this case, the total frequency in which rare actions occur in the samples might be too small. This is why having a less balanced game 
(i.e., condition (3) of $g$ not being too close to 0.5) induces the population to move away from a pure equilibrium.

\paragraph{Heterogeneous populations and stable symmetric mixed states} 
Our next result focuses on the action-sampling dynamics. We show that the behavior of the population is qualitatively different for homogeneous populations (in which all agents have the same sample size) than for heterogeneous populations. Specifically, we show (Theorem \ref{thm:unstableinterior_ASD}) that all homogeneous populations (with a sample size of at least 2) converge from almost all initial states to one of the pure states under the action-sampling dynamics. By contrast, Theorem \ref{thm:stableinterior_ASD} shows that the symmetric stationary  mixed state (in which both populations obtain the same expected payoff, and bargaining frequently fails) is asymptotically stable in a large class of heterogeneous populations in which some agents have relatively small samples, while the remaining agents have large samples. Interestingly, \emph{the symmetric mixed state is asymptotically stable even in populations in which most agents perfectly observe the opposing population's true distribution of play, while a small minority relies on small samples} (provided that $g$ is sufficiently far from 0.5).

The argument that the symmetric stationary state  is unstable in homogeneous  populations and stable in heterogeneous populations is sketched as follows. A symmetric stationary state under the action-sampling dynamics is characterized by solving the equation $p=\Pr(X(k,p)\leq{m})$, where $X(k,p)$ is a binomial random variable and $m\in \{0,1,\dots,k-1\}$. This is so because the rate of hawks ($p$) among the new agents is equal to the probability that the number of the opponents' hawkish actions observed by a new agent (which has a binomial distribution) is below some threshold ($m$). One can show that this solution is close to the peak of the unimodal density of $X(k,p),$ which implies that a small perturbation in $p$ has a strong effect on $\Pr(X(k,p)\leq{m})$, and thus it substantially changes the new agents' behavior, which takes the population away from the symmetric state. By contrast, in heterogeneous populations the stationary state falls in between the peaks of the two distributions, which implies that small perturbations have a weak effect on the agents' behavior.

Our final result (Theorem \ref{thm:stableinterior_PSD}) shows that the conditions for the stability of the symmetric mixed state are broader under the payoff-sampling dynamics, and that populations with relatively small sample sizes (either homogeneous or heterogeneous) can have an asymptotically stable symmetric mixed state.

Taken together, our results show that when 
some agents have limited information about the aggregate behavior of the population (which seems plausible in various real-life applications, such as the motivating example of buyers and sellers of houses), then an egalitarian, yet inefficient, convention may arise in which bargaining frequently fails. 

 Our proofs rely on deriving new properties of binomial distributions, which may be of independent interest. One 
 such new property (Proposition \ref{pro-binomial}) is that  there do
not exist any two different probabilities $p\neq q\in(0,1)$ such that the probability of obtaining at most $m$ successes out of $k$ trials is $q$ when the success probability in each trial is $p$, and at the same time, the probability of obtaining at most $m$ successes out of $k$ trials is $q$ when the success probability in each trial is $p$.

\paragraph{Structure} 
Section \ref{related-literature} presents the related literature. Our model is described in Section \ref{sec:model}. Section \ref{sec:local stability pure states} presents a ``complete'' characterization for global convergence to one of the mixed stationary states. In Section \ref{sec-conv-to-pure} we show that homogeneous populations converge to pure stationary states under the action-sampling dynamics.  Section \ref{sec:local stability mixed states} presents various sufficient conditions for the  stability of the  symmetric stationary state.
The analytic results of the paper are supplemented by a numeric analysis in Section \ref{sec:Numeric-Analysis}.
We conclude in Section \ref{sec:conclusion}. Formal proofs are presented in the appendix. 
Appendix \ref{sub:binomial-proofs} describes our results for binomial distributions.

\section{Related Literature}\label{related-literature}
\paragraph{Related theoretical literature}

The action-sampling dynamics were pioneered by 
  \cite{sandholm2001almost} and \cite{osborne2003sampling}. \cite{oyama2015sampling} applied these dynamics to prove global convergence results in supermodular games. Recently, \cite{heller2018social} studied the conditions on the expected sample size 
    that implies global convergence for all payoff functions and all sampling dynamics.\footnote{   \cite{hauert2018effects} use the term ``sampling dynamics'' to refer to a variant  of the replicator dynamics, in which when an agent samples another agent and mimics the other agent's behavior, it is more likely that these two agents will be matched with each other. This is less related to our use of the notion of ``Sampling dynamics'', which is in line with the literature cited in the main text.}

  \cite{salant2020statistical} (see also \citealp{sawa2021statistical}) generalized the action-sampling dynamics by allowing new agents to use various procedures to infer from their samples the aggregate behavior of the opponents 
  (in addition to allowing for payoff heterogeneity in the population).  \citeauthor{salant2020statistical} pay special attention to \emph{unbiased} inference procedures in which the agent's expected belief about the share of opponents who play hawk coincides with the sample mean. Examples of unbiased procedures are maximum likelihood estimation, beta estimation with a prior representing complete ignorance, and a truncated normal posterior around the sample mean. In our setup, the payoffs are linear in the share of agents who play hawk, which implies that the agent's perceived best reply depends only on the expectation of her posterior belief. This implies that our results hold for any unbiased inference procedure.
  
  The present paper, similar to the papers cited above, studies deterministic dynamics in infinite populations. When there is convergence to a stable stationary state in such dynamics, the convergence is fast
  (\citealp{oyama2015sampling}).
  By contrast, stochastic evolutionary models (see, e.g., the seminal contribution of \citealp{young1993evolution}, 
  and the recent hawk--dove application in \citealp{Bilancini-et-al-2021}), which are also based on revising agents observing a finite sample of opponents' actions,  focus on the very long-run behavior of stochastic processes when players’ choice rules include the possibility of rare “mistakes” 
  (sufficient conditions for stochastic evolutionary models to yield fast convergence are studied in \citealp{kreindler2013fast,arieli2020speed}).
 
The payoff-sampling dynamics were pioneered by \cite{osborne1998games} and \cite{sethi2000stability} and later generalized in various respects by \cite{sandholm2020stability}. It has been used in a variety of applications, including price competition with boundedly rational consumers (\citealp{spiegler2006competition}), common-pool resources (\citealp{cardenas2015stable}), contributions to public goods (\citealp{mantilla2018efficiency}), centipede games (\citealp{sandholm2019best}), finitely repeated games (\citealp{Raj}) and the prisoner's dilemma (\citealp{arigapudi2020instability}).
The existing literature assumes that all agents have the same sample size. 

A methodological contribution of the present paper is in extending the setup of payoff-sampling dynamics to analyze heterogeneous populations in which new agents differ in their sample sizes, and this heterogeneity leads to 
qualitatively new 
results.

It is well known that mixed stationary states in  multiple-population games cannot be asymptotically stable under the commonly used replicator dynamics (see, e.g.,  \citealp[Theorem 9.1.6]{sandholm2010population}).
 By contrast, we show that both classes of sampling dynamics, which are plausible in various real-life applications, can induce asymptotically stable mixed stationary states in a two-population hawk--dove game.

\paragraph{Related experimental literature}
\cite{selten2008stationary} experimentally tested the predictive power of various solution concepts in two-action, two-player games with a unique completely mixed Nash equilibrium. They show that 
both the payoff-sampling equilibrium and the action-sampling equilibrium outperform the predictions of both the Nash equilibrium and the quantal-response equilibrium.

Recently, \cite{stephenson2019coordination} tested the predictive validity of various evolutionary models in coordinated attacker–defender games.\footnote{Experiments that directly test the dynamic predictions of evolutionary game theory are quite scarce. Two notable exceptions are the experiments showing the good fit of the dynamic predictions in the rock--paper--scissors game (\citealp{cason2014cycles,hoffman2015experimental}).} \citeauthor{stephenson2019coordination}'s experimental design is very favorable for monotone dynamics because each participant is shown the exact (population-dependent) payoff that would be obtained by each action at each point in time. Nevertheless, subjects frequently violate monotonicity: 10\%--20\% of the subjects switch from higher-performing strategies to lower-performing strategies.

The key prediction of monotone dynamics for hawk–dove games (in which agents from one population are randomly matched with agents from another population) is experimentally tested in \cite{oprea2011separating} and \cite{benndorf2016equilibrium}. Both experiments apply an interface that is favorable to monotonicity (i.e., each participant is shown the exact population-dependent payoff of each action). Both experiments show that the prediction of monotone dynamics holds in this setup, and that the populations converge to an asymmetric pure equilibrium in which one population (say, the buyers) plays hawk and the other population (say, the sellers) plays dove. 

Consider a revised experimental design, where an agent observes only the behavior of her own opponent, rather than the aggregate behavior of the opposing population. An interesting testable prediction of our model is that in this experimental design, the populations are likely to converge to the symmetric stationary state in the relevant parameter domain 
(in particular, when $g$ is not too far from 1; see Figure
\ref{fig:numeric_figure}).\footnote{\cite{benndorf2016equilibrium, benndorf2021games} 
studied a more general setup in which each participant in each round is randomly matched with an opponent from the other population with probability $\kappa \in (0,1)$, and is randomly matched with an opponent from her own population with the remaining probability $1-\kappa$. Our theoretical  predictions fit the setup of $\kappa$ close to one.}

\section{Model}\label{sec:model}

\subsection{The Hawk--Dove Game}
Let $G=\{A,u\}$ denote a symmetric two-player hawk--dove game, where: 
\begin{enumerate}
\item $A=\{h,d\}$
is the set of actions of each player, and 
\item $u:A^2\rightarrow\mathbb{R}$
is the payoff function of each player.
 \end{enumerate}
Let $i\in\left\{ 1,2\right\} $ be an index referring to one of the
players, and let $j=\left\{ 1,2\right\} \backslash\left\{ i\right\} $
be an index referring to the opponent. We interpret action $h$
as the hawkish (more aggressive) action and $d$ as the dovish
action. The payoff matrix $u(\cdot,\cdot)$ of a generalized hawk--dove game is given in Table \ref{tab:Generalized-Hawk--dove-Games}.
When both agents are dovish, they obtain a relatively high payoff, which is normalized to 1. When both agents are hawkish, they obtain their lowest feasible payoff, which is normalized to 0. 
Finally, when one of the players is hawkish and her opponent is dovish, the hawkish player \emph{gains}  $g\in(0,1)$ (relative to the payoff 1 obtained by two dovish players), while her dovish opponent \emph{loses}\footnote{\cite{herold2020evolution} allow a broader 
domain in which the assumption of $g,l\in(0,1)$ is replaced with the weaker assumption of
$g>0$, $l<1,$ and $l+g>0$. All of our results hold in this extended setup.
} $l\in(0,1)$. 
The game admits three Nash equilibria: two asymmetric pure Nash equilibria: $(h,d)$ and $(d,h),$ and a symmetric mixed Nash equilibrium in which each player plays $h$ with probability $\frac{g}{1+g-l}$,
and obtains a relatively low expected payoff of $\frac{(1+g)(1-l)}{1+g-l}<1$.

\begin{table}
\begin{centering}
\caption{Payoff Matrix of a Generalized Hawk--Dove Game $g,l\in\left(0,1\right)$ \label{tab:Generalized-Hawk--dove-Games}}
\par\end{centering}
\smallskip{}
\centering
    \setlength{\extrarowheight}{2pt}
    \begin{tabular}{cc|c|c|}
      & \multicolumn{1}{c}{} & \multicolumn{2}{c}{\textcolor{red}{Player 2} \vspace{-0.5em}}\\
      & \multicolumn{1}{c}{} & \multicolumn{1}{c}{$\textcolor{red}{h}$}  & \multicolumn{1}{c}{$\textcolor{red}{d}$} \\\cline{3-4}
      \multirow{2}*{\textcolor{blue}{Player 1}}  & $\textcolor{blue}{h}$ & $\textcolor{blue}{0},\textcolor{red}{0}$ & $\textcolor{blue}{1+g},\textcolor{red}{1-l}$ \\\cline{3-4}
      & $\textcolor{blue}{d}$ & $\textcolor{blue}{1-l},\textcolor{red}{1+g}$ & $\textcolor{blue}{1},\textcolor{red}{1}$ \\\cline{3-4}
    \end{tabular}
  \end{table}
An important special 
subclass  is
the \emph{standard} hawk--dove games, 
in which $g=l$ (see Table \ref{tab:canonical-Hawk--dove-Games-1}); i.e., the gain of the hawkish player is equal to the loss of her dovish opponent.
\subsection{Evolutionary Process}

We assume that there are two unit-mass continuums of agents 
(e.g., buyers and sellers) and that
agents in population 1 are randomly matched with agents in population
2. Aggregate behavior in the populations
at time $t\in\mathbb{R}^{+}$ is described by a \emph{state}
$\mathbf{p}\left(t\right)=\left(p_{1}\left(t\right),p_{2}\left(t\right)\right)\in\left[0,1\right]^{2}$,
where $p_{i}\left(t\right)$ represents the share of agents playing
the hawkish action $h$ at time $t$ in population $i$. We extend
the payoff function $u$ to states (which have the same
representation as mixed strategy profiles) in the standard linear
way. Specifically, $u(p_i,p_j)$ denotes the average payoff of population $i$ (in which a share $p_i$ of the population plays $h$)   
when randomly matched against population $j$ (in which a share $p_j$  plays $h$).
With a slight abuse of notation, we use $d$ (resp., $h$)
to denote a degenerate population in which all of its agents play action $d$ (resp., $h$).
A state $\mathbf{p}=\left(p_{1},p_{2}\right)$ is \emph{symmetric} if $p_{1}=p_{2}$. 
A state $\mathbf{p}=\left(p_{1},p_{2}\right)$ is 
\emph{mixed} (or \emph{interior}) if $p_1,p_2\in(0,1).$ 

Agents occasionally die and are replaced by new agents (equivalently,
agents occasionally receive opportunities
to revise their actions).
Let $\delta>0$ denote the death rate
of agents in each population, which we assume to be independent of
the currently used actions. It turns out that $\delta$ does not have any effect on the dynamics,
except to multiply the speed of convergence by a constant. 
The evolutionary process is represented by a continuous function
$\textbf{w}:\left[0,1\right]^{2}\rightarrow\left[0,1\right]^{2}$, which describes
the frequency of new agents in each population who adopt action $h$ as a function of the current state. That is, $w_{i}\left(\mathbf{p}\right)$
describes the share of new agents of population $i$ who adopt action
$h$, given state $\mathbf{p}$. 
Thus, the instantaneous change in the share of agents of population $i$ that play hawk is given by 
$\dot{p}_{i}=\delta\cdot\left(w_{i}\left(\mathbf{p}\right)-p_{i}\right).\label{eq:general-dyanmics}$

\begin{rem}\label{rem-two-one}
Our \emph{two-population} dynamics fit situations in which each player can condition
her play on her role in the game (being Player 1 or Player 2). Common examples of such situations are (1) when sellers are matched with buyers, as in the motivating example, and (2) when each player observes if she has arrived slightly earlier or slightly later at a contested resource (\citealp{smith1982evolution}).
The two-population dynamics essentially induce the same results as one-population dynamics over a larger game with  $2\times2$ “role-conditioned” actions 
(see, e.g., \citealp{weibull1997evolutionary}, end of Section 5.1.1): being a hawk in both roles, being a dove in both roles, being a hawk as Player 1 and a dove as Player 2 (the “bourgeois” strategy of \citealp{smith1982evolution}), and being a hawk as Player 2 and a dove as Player 1.

By contrast, in \emph{one-population} dynamics 
of the original two-action game, an agent cannot condition her play on her role. It is well known that all monotone one-population dynamics converge to the unique mixed Nash equilibrium in hawk--dove games 
(see, e.g., \citealp[Section 4.3.2]{weibull1997evolutionary}). It is relatively straightforward to establish that one-population sampling dynamics lead to qualitatively similar results (convergence is to a somewhat different interior state than in
the mixed Nash equilibrium, but the comparative statics with respect to the payoff parameters remain similar).
\end{rem}
\paragraph{Monotone Dynamics}
The most widely studied 
dynamics are those that are monotone with respect to the payoffs.
A dynamic is monotone if the share
of agents playing an action increases iff the action yields a higher
payoff than the alternative
action.\footnote{In games with more than two actions, there are various definitions that capture different aspects of monotonicity. All these definitions coincide for two-action games. In particular, Definition \ref{def-monotone} coincides in two-action games with \citeauthor{weibull1997evolutionary}'s (\citeyear[Section 5.5]{weibull1997evolutionary}) textbook definitions of payoff monotonicity, payoff positivity, sign preserving, and weak payoff positivity.}$^{,}$\footnote{The best-known example of payoff monotone dynamics is the standard
replicator dynamic (\citealp{taylor1979evolutionarily}), which is
given by
$\dot{p}_{i}=w_{i}\left(\textbf{p}\right)-p_{i}=p_{i}\left(u\left(h,p_{j}\right)-u\left(p_{i},p_{j}\right)\right).\label{eq:replicator-dynamics}$}
\begin{defn}\label{def-monotone}
The dynamic $\textbf{w}:\left[0,1\right]^{2}\rightarrow\left[0,1\right]^{2}$
is \emph{monotone} if for any player $i$, any interior $p_{i}\in\left(0,1\right)$,
and any $p_{j}\in\left[0,1\right]$: $\dot{p}_{i}>0\,\,\Leftrightarrow\,\,u\left(h,p_{j}\right)>u\left(d,p_{j}\right).$ 
\end{defn}
\cite{oprea2011separating} showed that under monotone dynamics, from almost any initial state, the populations converge to one of the two asymmetric pure equilibria in which one population always plays
$h$ and the other population always plays $d$ (generalizing the seminal analysis of \citealp{smith1976logic}). 

\subsection{Sampling Dynamics}
In what follows, 
we study two plausible nonmonotone dynamics, in which new agents base their choice on inference from small samples. 

\paragraph{Distribution of sample sizes}

We allow heterogeneity in the sample sizes used by new agents. Let
$\theta\in\Delta\left(\mathbb{Z}_{+}\right)$ denote the distribution
of sample sizes of new agents. A share of $\theta\left(k\right)$
of the new agents have a sample of size $k$. Let
$\text{supp}\left(\theta\right)$ denote the support of $\theta$,  let 
$\max($supp$(\theta))$ 
denote the maximal 
sample size in the support of
$\theta$, and let $\max($supp($(\theta))=\infty$ if $\theta$'s
support is unbounded.
If there exists some $k,$ for which $\theta(k)=1,$ then we use
$k$ to denote the degenerate (homogeneous) distribution $\theta\equiv k$.

\begin{defn}
An \emph{environment} is a tuple $E=(g,l,\theta)$ where $g,l\in (0,1)$ describe the underlying hawk--dove game, and $\t\in\D\left(\mathbb{Z}_{+}\right)$ describes the distribution of sample sizes. 
\end{defn}

\paragraph{Action-sampling dynamics}
The action-sampling dynamics (\citealp{sandholm2001almost}) fit situations in which agents do not know the exact distribution of actions being played in the opponent's population, but know the payoffs of the underlying game. Agents estimate the unknown distribution of actions by sampling a few opponents' actions. Specifically, each new agent with sample size $k$ (henceforth, a $k$-agent) samples
$k$ randomly drawn agents from the opponent's population and then adopts the action that is the best reply  against the sample. 
\emph{To simplify notation, we assume that in case of a tie, the new agent plays action $d$}. Our results are
qualitatively the same for any tie-breaking rule.

Let ${X}(k,p_{j})\sim Bin\left(k,p_{j}\right)$
denote a random variable with binomial distribution with parameters
$k$ (number of trials) and $p_{j}$ (probability of success in each
trial), which is interpreted as the number of $h$-s in the sample. 
Observe that the sum of payoffs of playing action $h$ against the sample is $(1+g)\cdot(k-X(k,p_j))$ and the sum of payoffs of playing action $d$ against the sample is
$(k-X(k,p_j))+X(k,p_j)\cdot(1-l)=k-lX(k,p_j)$. 

This implies that action $h$ is the unique best reply to a sample of size $k$ iff $(1+g)(k-X(k,p_j)>k-lX(k,p_j)\Leftrightarrow \frac{X(k,p_j)}{k}<\frac{g}{1+g-l}$. This, in turn, implies that the action-sampling dynamic in environment $(g,l,\t)$ is given by
\begin{equation}
w_\theta^{A}\left(p_{j}\right)\equiv w_{i}^{A}\left(\textbf{p}\right) = \sum_{k\in \text{supp}\left(\theta\right)}\theta(k)\cdot\Pr\left(\frac{X\left(k,p_{j}\right)}{k}<\frac{g}{1+g-l}\right).\label{eq:action-sampling0dyanmics}
\end{equation}

\paragraph{Payoff-sampling dynamics}
The payoff-sampling dynamics (\citealp{osborne1998games}) fit situations in which agents either do not know the payoff matrix or do not have feedback about the actions being played in the opponent's population. Specifically, a new $k$-agent observes for each of her feasible actions the mean payoff obtained by playing this action in $k$ interactions (with each play of each action being against a newly drawn opponent), and then chooses the action whose mean payoff was highest. One possible interpretation for these observations is that each new $k$-agent tests each of the available actions $k$ times, and then adopts for the rest of her life the action with the highest mean payoff during the testing phase. As above, we assume that a tie induces a new agent to play  action $d$ (and the results are qualitatively the same for any tie-breaking rule). 

We refer to the sample against which action $h$ (resp.,
$d$) is tested as the $h$-sample (resp., $d$-sample). Let $X(k,p_j),Y(k,p_{j})\sim Bin\left(k,p_{j}\right)$
denote two iid random variables with a binomial distribution with parameters $k$ and $p_{j}$. Random variable $X(k,p_j)$ (resp., $Y(k,p_j)$) is interpreted as the number of times in which the opponents have played action $h$ in the $h$-sample (resp., $d$-sample). 
Observe that the sum of payoffs of playing action $h$ (resp., $d$) against its $h$-sample (resp., $d$-sample) is $(1+g)(k-X(k,p_j))$ 
(resp., $k-l Y(k,p_j$)).
This implies that action $h$ has the highest mean payoff  iff
$(1+g)(k-X(k,p_j))>k-lY(k,p_j)\Leftrightarrow (1+g)X(k,p_j)<gk+lY(k,p_j)$. Thus the payoff-sampling dynamic in environment $(g,l,\t)$ is given by
\begin{equation}
w_\theta^{P}\left(p_{j}\right)\equiv w_i^{P}\left(\textbf{p}\right) =\sum_{k\in \text{supp}\left(\theta\right)}\theta(k)\cdot\Pr((1+g)X(k,p_j)<gk+lY(k,p_j)).\label{eq:payoff-smapling-dynamics}    
\end{equation}
Henceforth, we omit the superscript $A$ / $P$ (i.e., we write $w_\theta(p_j)$) when referring to properties that hold for both action-sampling and payoff-sampling dynamics, 
or when it is clear from the context which class of sampling dynamics we are dealing with.

\section{Global Convergence to Mixed States}\label{sec:local stability pure states}
Recall that under monotone dynamics the population converges from almost any initial state to one of the pure equilibria.
In this section, we fully characterize the  conditions for which the opposite result holds under sampling dynamics; i.e., the populations converge from almost any initial state to one of the mixed (interior) stationary states.

\subsection{Analysis of $w(p_1)$ and Preliminary Results}
The dynamic characteristics of both classes of sampling dynamics are closely related to the properties of the polynomial 
$w_\theta(p_1)$, its inverse $w^{-1}_\theta(p_1),$ and their intersection points, which are analyzed in this subsection.

Figure \ref{figure-generic} illustrates the phase plots of the sampling dynamics and the properties of the polynomial $w_\theta(p_1)$ and its inverse $w_\theta^{-1}(p_1)$  for the environment in which $g=l=0.25$ and $\theta\equiv3$. The green solid curve is the polynomial $p_2=w_\theta(p_1)$, which describes the states in which $\dot p_1=0$. The orange dashed curve is the polynomial $p_2=w_\theta^{-1} (p_1).$ 

\begin{figure}
\begin{centering}
 \caption{Illustrative Phase Plots ($g=l=0.25$ and $\theta\equiv3$)  \label{figure1}}
 \includegraphics[scale=0.6]{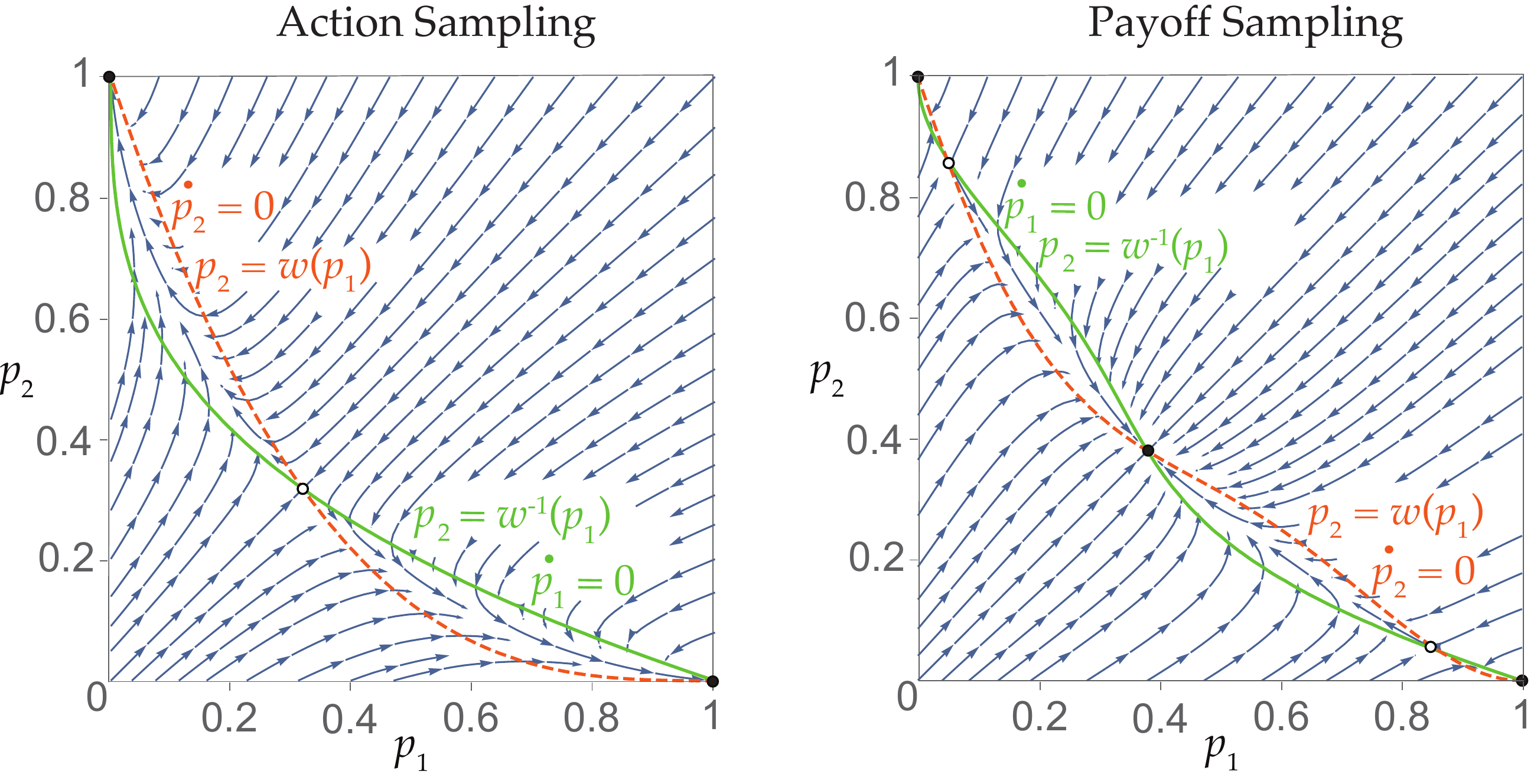}\label{figure-generic}
\par\end{centering}
    {\small{}
   The figure illustrates the phase plots of the action-sampling dynamics (left panel) and the payoff-sampling dynamics (right panel) for the environment  in which $g=l=0.25$ and $\theta\equiv3$. The green solid (resp., orange dashede) curve shows the states for which $\dot p_1=0$. The intersection points of these curves are the stationary states. A solid (resp., hollow) dot represents an asymptotically stable (resp.,  unstable) stationary state.}{\small\par}
\end{figure}

 The following fact is immediate (for both classes of sampling dynamics) from  the basic properties of  binomial random variables.
\begin{fact}\label{fact-w}
$w_\theta(p_j)$ is a 
strictly decreasing polynomial function that satisfies $w_\theta(0)=1$ and $w_\theta(1)=0$.
This implies that the inverse function $w_\theta^{-1}:[0,1]\rightarrow[0,1]$ exists, is continuously differentiable, and that $w_\theta^{-1}(0)=1$ and $w_\theta^{-1}(1)=0.$
\end{fact}
Fact \ref{fact-w} implies that the two curves intersect at $(0,1)$, $(1,0)$ and at a unique symmetric state $(p,p)$. The curves might admit additional intersection points.

Appendix \ref{sub-standard-definitions} presents the standard definitions of  stationary states, asymptotically stable states, and unstable states. Observe that a state is stationary
(i.e., it is a fixed point of the dynamics) iff it is an intersection point of the two curves $w_\theta$ and $w_\theta^{-1}$.

\begin{fact}\label{fact-intersection}
State $(p_1,p_2)$ is stationary iff $w_\theta(p_1)=w_\theta^{-1}(p_1)=p_2$.
\end{fact}

Our first result shows that the population converges to one of the stationary states from any initial state.\footnote{
One could also prove Proposition \ref{pro-convergence-to-stationary} by relying on the Bendixson--Dulac theorem (see Theorem 9.A.6 of \citealp{sandholm2010population}). We present a direct proof as the arguments are helpful in the proof of Proposition \ref{pro-convergence-to-pure-iff-asymp}.}
\begin{prop}\label{pro-convergence-to-stationary}
  $\lim_{t\rightarrow\infty}\mathbf{p}\left(t\right)$ exists for any $\textbf{p}(0),$ and it is a stationary state.
\end{prop}
\begin{proof}[Sketch of Proof]
Observe that in any trajectory that begins above (resp., below) both curves, the share of hawks in both populations decreases (increases), which implies that the trajectory moves downward and to the left (resp., upward and to the right) until it intersects with one of the curves. Thus, we can focus on trajectories that begin between the two curves. Assume w.l.o.g. that the initial state is below 
 $w(p_1)$ and above  $w^{-1}(p_1)$ (the argument in the opposite case is analogous). In any such state the share of hawks in population 1 (resp., 2) increases (resp., decreases). This implies that the trajectory moves upward and to the left until meeting one of the curves. This meeting point must be an intersection point of both curves, because if the meeting point were only with $w(p_1)$ (resp., $w^{-1}(p_1)$), the trajectory there would have been  horizontal to the left (vertical upward), which implies that the trajectory would come from the right side of the curve (from below the curve), leading to a contradiction.  See Appendix \ref{proof-convergence-to-stationary} for a formal proof.
\end{proof}
Our second result shows that if any initial state converges to one of the pure equilibria, then this equilibrium must be asymptotically stable (as defined in Appendix \ref{sub-standard-definitions}).
\begin{prop}\label{pro-convergence-to-pure-iff-asymp}
Assume that  $\textbf{p}(0)\neq (1,0)$ and $\lim_{t\rightarrow\infty}\mathbf{p}\left(t\right)=(1,0)$; then $(1,0)$ is asymptotically stable. The same result holds when replacing $(1,0)$ with $(0,1).$
\end{prop}
\begin{proof}[Sketch of Proof]
Due to analogous arguments to the above sketch of proof we can assume w.l.o.g. that $\textbf{p}(0)$ is between the two curves, and that the trajectory moves upward and to the left (which is required for convergence to (1,0)) iff the curve $w(p_1)$ is above $w^{-1}(p_1)$, and the closest intersection point of the two curves to the left of  $\textbf{p}(0)$ is (1,0). These two conditions imply that any initial state sufficiently close to (1,0) converges to (1,0), which implies that (1,0) is asymptotically stable.
See Appendix \ref{proof-converge-pure-iss-asymp} for the formal proof.
\end{proof}

\subsection{Single Appearance of a Rare Action}
The following lemma characterizes when a single appearance of a rare action in a new agent's sample can change the agent's behavior (see the proof in Appendix \ref{lem-proof}; 
we require strictly higher payoffs for action $h$ and weakly higher payoffs for action $d$ due to our tie-breaking rule in favor of action $d$).

\begin{lem}
\label{lemma-sample} Consider a new agent in population $i$ with a sample size of $k$.
\begin{enumerate}
\item Action-sampling dynamics: (I) Action $h$ induces a 
strictly higher payoff
against a sample with a single opponent's action $d$ iff
$k<\frac{1+g-l}{1-l}$ 
and (II) Action $d$ induces a
weakly higher payoff against a sample
with a single $h$ iff $k\leq \frac{1+g-l}{g}$. 
\item Payoff-sampling dynamics: (I) a $h$-sample with a single $d$
induces a
strictly higher mean payoff than a $d$-sample with no $d$-s
iff $k<\frac{1+g}{1-l}$ and (II) a $d$-sample with no $h$-s
induces a weakly higher mean payoff than an $h$-sample with a single
$h$ iff $k\leq \frac{1+g}{g}$. 
\end{enumerate}
\end{lem}

Lemma \ref{lemma-sample} allows us to define the 
upper bounds on the sample size in which a single appearance of a rare action can change the behavior of a new agent.

\begin{defn}\label{def:maximal-samples}
Let 
 $\,m_{h}^{A}= 1+\frac{g}{1-l} ,\,m_{d}^{A}= 1+\frac{1-l}{g},
 m_{h}^{P}= \frac{1+g}{1-l} ,\,m_{d}^{P}= \frac{1+g}{g}.$
 
\end{defn}
    We omit the superscript ${A,P}$ when stating a result that is true for both dynamics; for example, we write $m_h$, which denotes $m_h^P$ when the underlying dynamics is payoff sampling, and which denotes $m_h^A$ when the underlying dynamics is action sampling. 
    
    The parameter $m_{h}$ is the 
    upper bound on the sample size for which
a single appearance of $d$ in the sample, when all other sampled actions are $h$, induces a new agent to adopt action $h$. Similarly,  $m_{d}$ is the 
upper bound on the sample size for which a single appearance of $h$ in the sample, when all other sampled actions are $d$, can induce a new agent to adopt action $d$. 
 
We conclude this subsection by presenting a definition of $m$\emph{-bounded expectation} of a probability distribution with support on the set of positive integers. It is the expected value of
the probability distribution by restricting its support to $m.$ Formally,

\begin{defn}
The $m$\emph{-bounded expectation} $\mathbb{E}_{\leq m}$ (resp., $\mathbb{E}_{< m}$) of distribution $\theta$ with support on
integers is\footnote{
Observe that in our notation the parameter $k$ takes only (positive) integer values (although we allow the upper bound $m$ to be a non-integer).}
$\mathbb{E}_{\leq m}\left(\theta\right)=\sum_{1 \leq k\leq m}\theta\left(k\right)\cdot k$ (resp., $\mathbb{E}_{< m}\left(\theta\right)=\sum_{1\leq k< m}\theta\left(k\right)\cdot k$).
\end{defn}
\subsection{Asymptotic Stability of Pure Equilbria}
Our next result characterizes the asymptotic stability of the pure
states. It shows that the asymptotic stability depends only on whether the product of the bounded expectations of the distribution of sample sizes is larger or smaller than one, where the bound of each distribution is the maximal sample size for which a single appearance of a rare action can change the behavior of a new agent. Formally 
(where replacing the tie-breaking rule with an $h$-favorable one would replace the ``<''-s and the ``$\leq$''-s in the statement):
\begin{prop}
\label{prop:pure-equilbiria-satbility}
\begin{enumerate}
\item $\mathbb{E}_{< m_{h}}\left(\theta\right)\cdot \mathbb{E}_{\leq m_{d}}\left(\theta\right)>1\,\Rightarrow\,$ both pure stationary states are unstable.
\item $\mathbb{E}_{<m_{h}}\left(\theta\right)\cdot \mathbb{E}_{\leq m_{d}}\left(\theta\right)<1\,\Rightarrow\,$ both pure stationary states are asymptotically stable. 
\end{enumerate}
\end{prop}

\begin{proof}[Sketch of Proof]
Consider a slightly perturbed state $(\epsilon,1-\epsilon)$ near the pure equilibrium $(0,1).$ Observe that almost all agents in population 1 (resp.,  2) play $d$ (resp., $h$). 
The event of two rare actions appearing in a sample of a new agent has a negligible probability of $O(\epsilon^2)$. If a new agent has a sample size of $k$, then the probability of a rare action appearing in the sample is approximately $k \cdot \epsilon.$ This rare appearance changes the perceived best reply of a new agent of population 1 iff $k$ is smaller than $m_{h}$. Thus, the total probability that a new agent of population 1 (resp., 2) adopts a rare action 
is equal to $\mathbb{E}_{<m_{h}}\left(\theta\right)$ (resp., $\mathbb{E}_{\leq m_{d}}\left(\theta\right)$). This implies that the product of the share of new agents adopting a rare action in each population is $\epsilon\cdot\mathbb{E}_{< m_{h}}\left(\theta\right)\cdot\e\cdot\mathbb{E}_{\leq m_{d}}\left(\theta\right)$. This shows that the share of agents playing rare actions gradually increases (resp., decreases) if $\mathbb{E}_{< m_{h}}\left(\theta\right)\cdot\mathbb{E}_{\leq m_{d}}\left(\theta\right)>1$ (resp., $\mathbb{E}_{< m_{h}}\left(\theta\right)\cdot\mathbb{E}_{\leq m_{d}}\left(\theta\right)<1),$ which implies instability (resp., asymptotic stability). See Appendix \ref{subsec:Proof-of-Theorem-2} for a formal proof.
\qedhere

\end{proof}
Observe that the fact that $l>0$ immediately implies that $m_{h}^{A}<m_{h}^{P}$ and $m_{d}^{A}<m_{d}^{P}$,
which, in turn, implies that instability under the action-sampling dynamics holds in a strictly smaller set of distributions than under the payoff-sampling dynamics. 
\begin{cor}
~\label{cor:Instability-A-instability-P}
If the pure stationary states are unstable under the action-sampling dynamics,
then they are also unstable under the payoff-sampling dynamics.

\end{cor}
Next we observe that the pure stationary states are stable under the action-sampling dynamics in populations in which all agents have sample sizes of at least 2.
\begin{cor}\label{cor:theta-1-action}
Assume that $\theta(1)=0$. The pure stationary states are stable under the action-sampling dynamics.
\end{cor}
\begin{proof}
Observe that either $\min(m_h^A,m_d^A)<2$ or $m_h^A=m_d^A=2$. Both conditions imply that one of the expressions in the product of $\mathbb{E}_{< m_{h}}\left(\theta\right)\cdot \mathbb{E}_{\leq m_{d}}\left(\theta\right)$ is equal to $\theta(1),$ which implies (by Proposition \ref{prop:pure-equilbiria-satbility}) that the pure stationary states are stable if $\theta(1)=0$. 
\end{proof}
\subsection{Main Result}
Combining Propositions \ref{pro-convergence-to-stationary}--\ref{prop:pure-equilbiria-satbility} yields the main result of this section. It shows that the population converges from almost any initial state to an interior stationary state if 
(and essentially only if) the product of the bounded expectations of the distribution of sample sizes is larger than 1.
\begin{thm}\label{thm:global-mixed}
\begin{enumerate}
    \item Assume that $\mathbb{E}_{< m_{h}}\left(\theta\right)\cdot \mathbb{E}_{\leq m_{d}}\left(\theta\right)>1.$ If $\boldsymbol{p}(0)\notin\left\{ (0,1),(1,0)\right\}$, then $\lim_{t\rightarrow\infty}\mathbf{p}\left(t\right)\in{(0,1)}^2$.
    \item Assume that $\mathbb{E}_{< m_{h}}\left(\theta\right)\cdot \mathbb{E}_{\leq m_{d}}\left(\theta\right)<1.$ Then there exist $\boldsymbol{p}(0),\boldsymbol{\hat{p}}(0)\notin\left\{ (0,1),(1,0)\right\}$ such that $\lim_{t\rightarrow\infty}\mathbf{p}\left(t\right)=(0,1)$ and $\lim_{t\rightarrow\infty}\mathbf{\hat{p}}\left(t\right)=(1,0)$. 
\end{enumerate}
\end{thm}
\begin{proof}
\begin{enumerate}
    \item Assume that  $\mathbb{E}_{< m_{h}}\left(\theta\right)\cdot \mathbb{E}_{\leq m_{d}}\left(\theta\right)>1.$ Proposition \ref{prop:pure-equilbiria-satbility} implies that both pure stationary states are unstable. Combining Propositions \ref{pro-convergence-to-stationary}--\ref{pro-convergence-to-pure-iff-asymp} implies that from almost any initial state, the population converges to an interior stationary state.
    \item Assume that  $\mathbb{E}_{< m_{h}}\left(\theta\right)\cdot \mathbb{E}_{\leq m_{d}}\left(\theta\right)<1.$ Proposition \ref{prop:pure-equilbiria-satbility} implies that both pure stationary states are stable, which implies part (2). \qedhere
\end{enumerate}
\end{proof}

Next we apply Theorem \ref{thm:global-mixed} to standard hawk--dove games in which the parameter $g=l$ 
describes both the gain of a hawkish player and the loss of her dovish opponent. 
In order to simplify the statement we assume that the various $m$-s are non-integers (to avoid the weak dependency on the tie-breaking rule).
\begin{cor}\label{cor:g=l}
Assume that $g=l$ and that $\frac{1}{g}$, $\frac{1}{1-g}$, $\frac{1+g}{g}$, and $\frac{1+g}{1-g}$ are non-integers. The population converges to a mixed stationary state from any interior state iff:
\begin{enumerate}
    \item Action-sampling dynamics: $\theta(1)\cdot\mathbb{E}_{\leq \max\left(\frac{1}{g},\frac{1}{1-g}\right)}\left(\theta\right)>1.$
    \item Payoff-sampling dynamics: $\begin{aligned}[t]
   \text{either } 
   (a) \ g&<\frac{1}{3} \ \text{and} \ \theta(1)\cdot\mathbb{E}_{\leq \frac{1+g}{g}}\left(\theta\right)>1\\
        \text{or } (b) \ g&\geq\frac{1}{3} \ \text{and}  \ (\theta(1)+2\theta(2))\mathbb{E}_{\leq \max\left(3,\frac{1+g}{1-g}\right)}\left(\theta\right)>1.                      \end{aligned}$
   \end{enumerate}
\end{cor}
The straightforward proof, which relies on substituting $g=l$ in Definition \ref{def:maximal-samples}, is presented in Appendix \ref{proof-cor-g=l}.
Corollary \ref{cor:g=l} implies that global convergence to one of the mixed stationary states holds iff: (1) sufficiently many agents have a sample size of 1, (2) the expected sample size 
(conditional on the sample being below a threshold that is increasing in $|g-0.5|$) is sufficiently large, and (3) $g$  is not too close to 0.5. Under the payoff-sampling dynamics (and assuming $g>\frac{1}{3}$), global convergence holds also if one relaxes Condition (1) by requiring that (1') sufficiently many agents have a sample size of at most 2.
As demonstrated in Section \ref{sec:Numeric-Analysis}, these conditions hold for many environments with 
relatively small samples under the payoff-sampling dynamics, and for a somewhat limited set of environments under the action-sampling dynamics.

\begin{rem}
Theorem \ref{thm:global-mixed} implies that if the sample sizes of all agents in the population are sufficiently large (larger than all the $m$-s in Definition \ref{def:maximal-samples}), then the pure stationary states are stable.
Perhaps surprisingly, a moderate increase in the sample sizes of some agents in the population can yield the opposite effect, that of destabilizing the pure stationary states. For example, when $g=l=0.25$, the pure stationary states are  stable under the action-sampling dynamics given a population in which $75\%$ of the agents have sample size 1 and $25\%$ have sample size 2 (because $\theta(1)\cdot\mathbb{E}_{\leq \max\left(\frac{1}{g},\frac{1}{1-g}\right)}\left(\theta\right)=\frac{3}{4}\cdot\frac{5}{4}=\frac{15}{16}<1$). By contrast, if one increases the samples of the latter group of agents from size 2 to size 3, then the pure stationary states become unstable (because $\theta(1)\cdot\mathbb{E}_{\leq \max\left(\frac{1}{g},\frac{1}{1-g}\right)}\left(\theta\right)=\frac{3}{4}\cdot\frac{3}{2}=\frac{9}{8}>1$); moreover, one can show that in this case the symmetric mixed state is globally stable.
\end{rem}
\section{Global Convergence to Pure States}\label{sec-conv-to-pure}
In this section, we show that \emph{homogeneous} populations (in which all agents have the same sample size) behave under the \emph{action-sampling dynamics} in the same way as under monotone dynamics: namely, they converge from almost any initial state to one of the pure equilibria. 
By contrast, In Section \ref{sec:local stability mixed states} we will show that 
convergence  to the symmetric mixed stationary state is achieved in many cases in which we have either: (1)  heterogeneous  populations or (2) payoff-sampling dynamics.

Theorem \ref{thm:unstableinterior_ASD} shows that if all agents have the same sample size, then the action-sampling dynamics admit a unique interior stationary state, which is unstable. This implies (together with Proposition \ref{pro-convergence-to-stationary}) that the phase plots of all such dynamics are qualitatively equivalent to the left panel of Figure \ref{figure1}, and that \emph{all initial populations eventually converge to one of the pure equilibria}
(in the sense that even if a population initially converges to the unique unstable interior equilibrium, an arbitrarily small perturbation will still take the population from there to one of the pure equilibria). 
\begin{thm}\label{thm:unstableinterior_ASD}
For\footnote{\label{foot-k-1}
It is straightforward to show that the case in which all agents have sample size 1 induces an environment in which the stationary states are $(p,1-p)$ for any $p\in[0,1]$, and that all these states are neither asymptotically stable nor 
unstable (i.e., they are Lyapunov stable) for any $g,l$ and under both sampling dynamics.} 
$\theta\equiv k>1,$ the action-sampling dynamics admit a unique interior stationary state, which is unstable.
\end{thm}

It turns out that Theorem \ref{thm:unstableinterior_ASD} is implied by a new general property of binomial distributions, which may be of independent interest. Proposition \ref{pro-binomial} states that there do not exist any two different probabilities $p\neq q\in(0,1)$ such that if the success probability in each trial is $p$, the probability of obtaining at most $m$ successes is $q$, and the same holds when swapping $p$ and $q$. 
This is formalized as follows. (Recall that $X(k,p)$ denotes a \emph{binomial random variable} with $k\geq1$ trials and probability of success $p$ in each trial.)

\begin{prop} \label{pro-binomial}
Fix any $k\geq2$ and any $0\leq m<k$. Then: 
\[
\Pr\left(X(k,p)\leq m\right)=q\in\left(0,1\right),\,\Pr\left(X(k,q)\leq m\right)=p\,\,\Rightarrow\,\, p=q.
\]
\end{prop}

The proof of Proposition \ref{pro-binomial} is presented in Appendix \ref{sub:binomial-proofs}.

\begin {proof}[Proof of Theorem \ref{thm:unstableinterior_ASD}] Each pair of parameters $g,l$ induces a threshold $0 \leq m < k$, such that playing hawk is the best reply against a sample of size $k$ iff the sample includes at most $m$ hawkish actions. That is, $w_k(p_1)$ must coincide with the curve of $\Pr\left(X(p_1,k)\leq m\right)$ for some $0 \leq m < k$.
This implies that an interior state $\textbf{p}$ is stationary iff $\Pr\left(X(p_1,k)\leq m\right)=p_2$ and $\Pr\left(X(p_2,k)\leq m\right)=p_1.$ Thus, Proposition \ref{pro-binomial} implies that there are no asymmetric interior stationary states. Corollary \ref{cor:theta-1-action} implies that the pure stationary states are asymptotically stable, and that the curve  $w(p_1)$ is above (resp., below) the curve $w^{-1}(p_1)$ near state $(1,0)$ (resp., $(0,1)$). This implies that curve $w(p_1)$ must be above (resp., below) the curve $w^{-1}(p_1)$ in a left (right) environment of the symmetric stationary state, which implies that this stationary state must be unstable, as illustrated in the left panel of Figure \ref{figure1}. Another, direct argument for that the symmetric stationary state is unstable is implied by combining Proposition \ref{pro-binom-deriv>1} ($\left|w'_k(p)\right|=\left|\frac{d\Pr\left(X(k,p)\leq m\right)}{dp}\right|>1$) and Proposition \ref{prop-symmetric-stablity} (state $(p,p)$ is asymptotically stable if $|w'_k(p)|>1$).  
\end{proof}

\section{Stability of the Symmetric Stationary State}\label{sec:local stability mixed states}
In this section, we present two sufficient conditions for the asymptotic stability of \emph{symmetric} mixed stationary states. Symmetric stationary states are qualitatively very different 
from pure stationary states. 
In particular, in symmetric states both 
populations have the same expected payoff, and 
the bargaining fails frequently 
(whenever two hawks are matched together). By contrast, \emph{asymmetric} mixed stationary states 
might be 
similar to pure stationary states in the sense of having $p_1$ close to 1, $p_2$ close to 0, and a large payoff difference between the two populations.
\subsection{Auxiliary Result for the Stability of  Symmetric States}
We begin with a simple auxiliary result showing that the symmetric stationary state is asymptotically stable iff the derivative of $w_\theta (p)$ is larger than one. Formally
\begin{prop}
Symmetric stationary state $(p,p)$ is unstable if $|w'_\theta (p)|>1$.\label{prop-symmetric-stablity}
\end{prop}
\begin{proof}
In order to assess the asymptotic stability, we compute the Jacobian $J$ of $\dot{p_i}=\delta\cdot(w_\theta(p_j)-p_i)$ at the symmetric rest point $(p, p)$ 
(ignoring the constant $\delta$, which plays no role in the dynamics, other than multiplying the speed of convergence by a constant): 
$$
J =  \left( \begin{array}{cc}
-1 & w'_\theta(p)  \\
w'_\theta(p) & -1  \end{array} \right).
$$
The eigenvalues of $J$ are $-1+w'_\theta(p)$ and $-1-w'_\theta(p).$ A sufficient condition for instability at $(p,p)$ is that there exists a positive eigenvalue, which is implied by $|w'_\theta(p)| > 1.$ 
\end{proof}
\subsection{Action-Sampling
Dynamics\label{subsec:Instability-of-symemtric-action-dyanmics}}
Our next result shows that the symmetric mixed stationary state is often asymptotically stable under the action-sampling dynamics if the population is heterogeneous, such that some agents have a relatively small samples, while the other agents have large samples.
Specifically, the symmetric stationary state is asymptotically stable for any population in which a positive, yet sufficiently small, share of the population have sample $k$, the remaining agents have sufficiently large samples, and either $g$ is sufficiently small or $l$ is sufficiently large.

\begin{thm}\label{thm:stableinterior_ASD}
Fix any $k>1$ and any $0<q<\frac{1}{k}$:
\begin{enumerate}
    \item For any $0<l<1$, there exists $0<\overline g<1$ and $\overline k\in \mathbb{N}$, such that the symmetric stationary state is asymptotically stable under the action-sampling dynamics for any $0<g<\overline g$ and any distribution $\theta$ satisfying $\theta(k)=q$ and $\sum_{n\geq\overline{k}}\theta(n)=1-q.$
     \item For any $0<g<1$, there exists $0<\overline l<1$ and $\overline k\in \mathbb{N}$, such that the symmetric stationary state is asymptotically stable under the action-sampling dynamics for any $\overline l<l<1$ and any distribution $\theta$ satisfying $\theta(k)=q$ and $\sum_{n\geq\overline{k}}\theta(n)=1-q.$
\end{enumerate}
\end{thm}

\begin {proof}[Sketch of Proof] 
To simplify the sketch, we treat the  agents with large samples ($n\geq\overline{k}$) as best responders (i.e., playing the best reply to the population's true distribution).
For sufficiently small $g$ (resp., sufficiently large $l$), the symmetric stationary state $p$ has the property that only agents with sample size $k$ play hawk (resp., dove), while all the agents with large samples play dove (resp., hawk), and they continue to do so even after a small perturbation in the share of hawks. Proposition \ref{prop-symmetric-stablity} 
implies that the symmetric stationary state $p$ is asymptotically stable if $|w'_\theta(p)|<1$. This latter inequality is implied by the fact that the best responders continue playing dove (resp., hawk) after any small perturbation and that $q$ is sufficiently small: $|w'_\theta(p)|=(1-q)\cdot0+q\cdot|w'_k(p)|<1.$ See Appendix \ref{subsec:stableinterior_ASDproof} for a formal proof. \qedhere
\end{proof}
\begin{figure}[h]
\centering
\caption{Illustrative Phase Plot for Theorem \ref{thm:stableinterior_ASD}: $g = l = 0.95,$ $\theta_2=0.3,$ and $\theta_{50}=0.7$}
\includegraphics[scale=0.75]{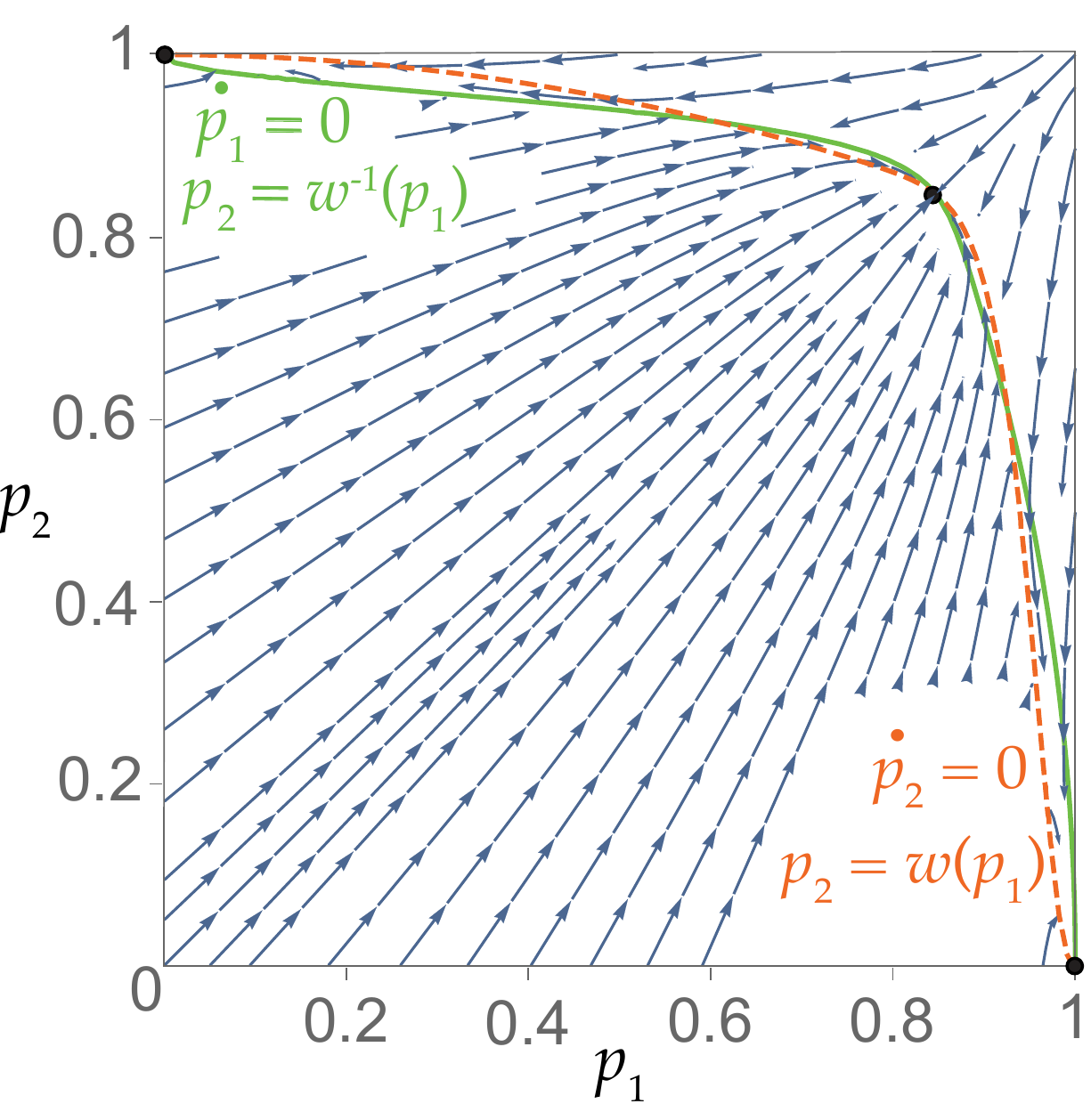}
\label{fig:fig-conv-to-sym}
\end{figure}
Figure \ref{fig:fig-conv-to-sym} illustrates the phase plot of a heterogeneous population that satisfies the conditions of Theorem \ref{thm:stableinterior_ASD}: the games is ``hawk-favorable'' with parameter $g=l=0.95$, $70\%$ of the agents have large samples of size 50, while the remaining $30\%$ of the agents have small samples of size 2. Numeric analysis shows that the basin of attraction of the symmetric stationary state in this example covers approximately $88\%$ of the unit square.  

It is straightforward to adapt Theorem \ref{thm:stableinterior_ASD} to show that the stationary symmetric state is asymptotically stable in various heterogeneous populations for the payoff-sampling dynamics as well. We omit the formal statement and proof for brevity.

Observe that
Theorem \ref{thm:stableinterior_ASD} implies that \emph{introducing a small share of agents with a finite sample size can stabilize the symmetric equilibrium of a population of (exact) best responders, as long as either $g$ is sufficiently small or $l$ is sufficiently large.}

\subsection{Payoff-Sampling
Dynamics\label{subsec:Stability-of-symmetric-payoff-dynamics}}
 In this subsection, we show that the symmetric stationary state is asymptotically stable under the payoff-sampling dynamics for homogeneous populations (and populations in which all agents have relatively small sample sizes).  For tractability in the analysis of payoff-sampling dynamics, we focus on the cases where the gain of a hawkish player and the loss of her dovish opponent are large, namely, $l,g>\frac{1}{\max (\textrm{supp}(\theta))}$. Our result shows that in this domain, the symmetric stationary state 
is asymptotically
stable in various populations in which agents have relatively small samples:
\begin{enumerate}
\item for any homogeneous  distribution of sample sizes $\theta\equiv k<20$; or
\item for any heterogeneous distribution with a maximal sample size of at most 5.
\end{enumerate}
The threshold of $k=20$ is binding. The symmetric stationary state becomes unstable if the sample size is $k \geq 20.$ By contrast, the bound of a maximal size of 5 for heterogeneous distributions of sample sizes is only a constraint of our proof technique. Numeric analysis (see Section \ref{sec:Numeric-Analysis}) suggests that the stability of the symmetric stationary state:
\begin{enumerate}
    \item holds for many distributions of types with larger maximal sample sizes (in particular, it holds for uniform distributions over 
$\left\{ 1,\dots,k\right\} $ for any $k\leq20$); and
    \item is often global (i.e., in many cases almost all initial states converge to this state).
\end{enumerate}
\begin{thm}\label{thm:stableinterior_PSD}
Assume that $l,g\in\left(\frac{1}{\max\left(\emph{supp}(\theta\right))},1\right)$, and either (1) $\theta\equiv k<20$, or (2) $\max\left(\emph{supp}(\theta\right))\leq5.$ Then, the game admits an asymptotically stable symmetric stationary state 
under the payoff-sampling dynamics. 
\end{thm}
\begin {proof}[Sketch of Proof] When $l$ and $g$ are sufficiently large, the payoff of action $h$ is slightly below twice the number of $d$-s in the $h$-sample, and the payoff of action $d$ is slightly above the number of $d$-s in the $d$-sample. This implies that action $h$ has a higher mean payoff than action $d$ iff the number of $d$-s in the $h$-sample is strictly greater than half the number of $d$-s in the $d$-sample. 
\begin{figure}[h]
\centering
\caption{The Function $w_k(p)$ for Various Values of $k$}
\includegraphics[scale=0.75]{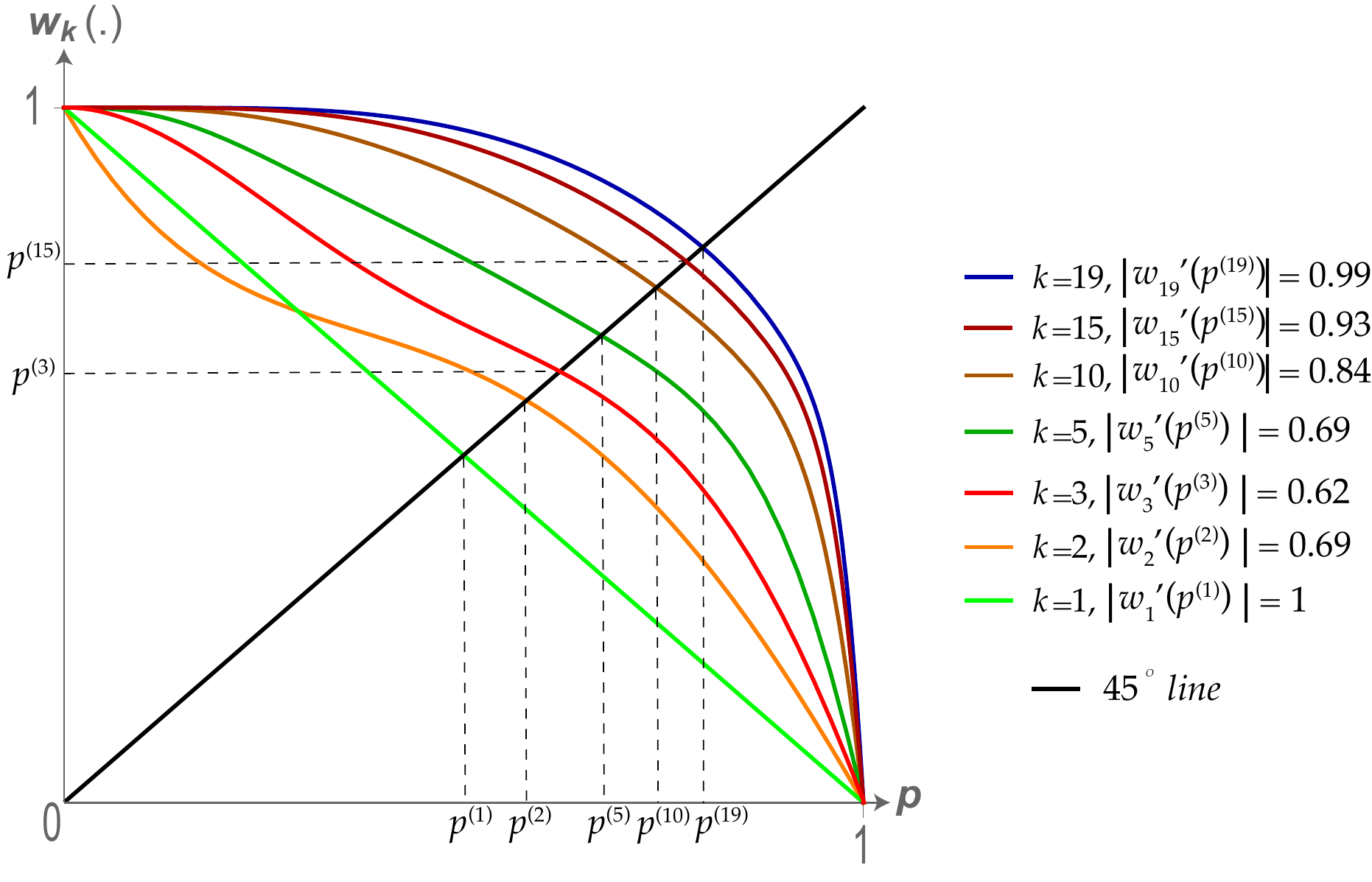}
\label{fig:proofbypicture}
\end{figure}

Thus, we can write $w_k(p)$ as follows:
\begin{equation}\label{eq:wkp-p}
w_{k}(p)=P\left(\underbrace{k-X(k,p)}_{\#d_{j}\,\textrm{in}\,h\textrm{\textrm{-}sample}}>\frac{1}{2}\underbrace{(k-Y(k,p))}_{\#d_{j}\,\textrm{in}\,d\textrm{-sample}}\right)=P(2X(k,p)-Y(k,p)<k),
\end{equation}
where $X(k,p)$ and $Y(k,p)$ are iid binomial random variables with parameters $k$ and $p.$ 

In the formal proof (see Appendix \ref{subsec:stableinterior_PSDproof}), we show that for any $k<20,$ $w_k(p)$ has a unique fixed point $p^{(k)}$ such that $|w_k'(p^{(k)})|<1$ (see Figure \ref{fig:proofbypicture}). This implies, by the same argument as in the sketch of proof of Theorem \ref{thm:unstableinterior_ASD}, that the symmetric stationary state is asymptotically stable. (By contrast, one can verify that $|w_k'(p^k)|>1$ for $k \geq 20,$ which implies that the symmetric stationary state is unstable for large $k\geq20$.) 

Next, we verify in the formal proof that for any $k \in \left\{1,2,3,4,5\right\}$ it holds that (I) the fixed points  are all in the interval $(0.5,0.68)$, and  (II) $|w_k'(p)|<1$ for any  $k \in \left\{1,...,5\right\}$ and any $p \in (0.5,0.68).$ Let $\theta$ be any distribution with $\max(\textrm{supp}(\theta)) \leq 5$. The fact that $w_\theta(p)$ is a weighted average of the various $w_k(p)$ implies that  (I) the fixed point $p^{(\theta)}$ of $w_\theta(p)$ is in $(0.5,0.68)$, and (II) $|w_\theta'(p^{(\theta)})|<1$ $\Rightarrow$ $(p^{(\theta)}, p^{(\theta)})$ is asymptotically stable.
\qedhere
\end{proof}

\section{Numeric Analysis\label{sec:Numeric-Analysis}}

We present numeric results that complement
the analytic results of the previous sections.

\paragraph{Methodology and Parameter Values}

The analysis focuses on the standard hawk--dove games, in which the gain of a hawkish player is equal to the loss of her dovish opponent, i.e., $g=l$
for each $i\in\left\{ 1,2\right\} $. We have tested the following $360=10 \times 27$ combinations of parameter values for each of the two sampling dynamics:
\begin{enumerate}
\item 10 values for $g$: 0.05, 0.15,
0.25, 0.35, 0.45, 0.55, 0.65, 0.75, 0.85, 0.95.
\item 36 distributions of sample sizes:
\begin{enumerate}
\item 9 \emph{homogeneous populations}, in which all agents have sample size $k$,
for each 
$k\in\{2,3,4,5,7,10,15,20,30\}$ (the case of $k=1$ is discussed in Footnote \ref{foot-k-1}). 
\item 9 \emph{uniform distributions} over $\left\{ 1,..,\overline{k}\right\} $, for each $\overline{k}\in\{2,3,4,5,7,10,15,20,30\}$.

\item 9 distributions with support $\{1,5\}$,
in which a share $q \in \left\{10\%, 20\%, ..., 90\%\right\}$ have sample size 1, and the remaining agents have sample size 5.
\item 9 distributions with support $\{2,30\}$,
in which a share $q \in \left\{10\%, 20\%, ..., 90\%\right\}$ have sample size 2, and the remaining agents have sample size 30.
\end{enumerate}
\end{enumerate}
For each set of parameters, we have numerically calculated the phase portrait and
the curves for which $\dot{p}{}_{1}=0$ and $\dot{p}{}_{2}=0$, and
used this to determine the dynamic behavior. 
The code is provided in the online supplementary material.\footnote{
Our numeric analysis is based on deterministic dynamics in a continuum population. We have randomly chosen 10 of these combinations of parameter values, and tested each of them by running it 100 times in the stochastic dynamics induced by  a finite population of 1,000 agents, using ABED software (\citealp{izquierdo2019introduction}). The results for finite populations are qualitatively the same.}
\begin{figure}\caption{Summary of Results of the Numeric Analysis}\label{fig:numeric_figure}
\begin{center}
\includegraphics[scale=0.67]{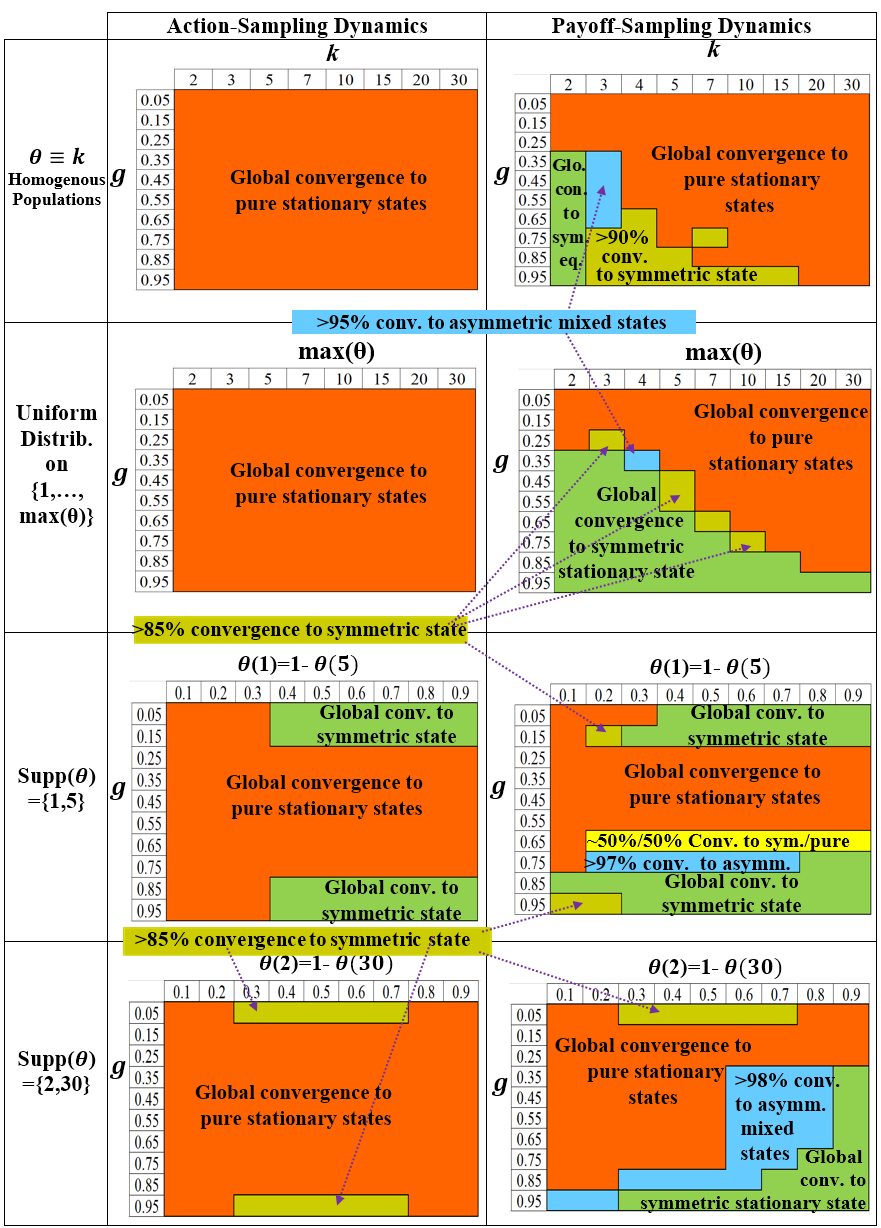}
 \end{center}
\end{figure}

\paragraph{Results}
The numeric results are summarized in Figure
\ref{fig:numeric_figure}. 
The \emph{action-sampling dynamics}  yield global convergence to the pure stationary states (orange shaded region in Figure
\ref{fig:numeric_figure}) in $90\%$ of the cases. In the remaining $10\%$ of the cases, the dynamics globally (or almost globally) converge to the symmetric stationary state (green shaded region), where these cases involve heterogeneous populations (some agents have small samples, and others have relatively large samples) and values of $g$ that are far from 0.5. 

The \emph{payoff-sampling dynamics} yield global convergence to the pure stationary states in about $50\%$ of the cases. In about $40\%$ of the cases the payoff-sampling dynamics globally (or almost globally) converge to the symmetric stationary state, where this occurs (1) in all cases in which the action-sampling dynamics converge to the symmetric stationary state, and (2) in many cases with large values of $g$. Finally, in the remaining $10\%$ of the cases the dynamics either converge to asymmetric mixed states (blue shaded region in Figure \ref{fig:numeric_figure}), or that both the pure stationary states and the symmetric stationary state have sizable basins of attractions (yellow shaded region).

\section{Conclusion}\label{sec:conclusion}
A key result in evolutionary game theory is that two populations that are matched to play a hawk--dove game converge to one of the pure equilibria from almost any initial state. We demonstrate that this result crucially depends on the revision dynamics being monotone. Specifically, we show that two plausible classes of  dynamics, in which
some agents base their chosen actions on sampling the actions of a few agents in the opponent population (action-sampling dynamics) or on sampling the payoffs of a few agents in their own population (payoff-sampling dynamics) can lead to the opposite prediction: convergence to a symmetric mixed stationary state. 

Our results provide a new explanation of why in bargaining situations, such as the motivating example of buying and selling houses, players in both populations tend to play hawkish strategies, and bargaining frequently fails.
Our model assumes that all players in each population have the same payoff matrix. Heterogeneity in the payoffs, and private information regarding one's payoff, are important aspects of many real-life bargaining situations. An interesting direction for future research is to apply the analysis of sampling dynamics in richer models that incorporate heterogeneous payoffs.  

\appendix
\section{Appendix\label{sec:Technical-Proofs}}
\subsection{Results for Binomial Random Variables}\label{sub:binomial-proofs}
In this section, we state and prove results for binomial random variables 
(result that to the best of our knowledge, do not appear in the existing literature). These results are used in  the proof of Theorem \ref{thm:unstableinterior_ASD}, and may also be of independent interest. Our first result shows that if the success probability $p$ in a single trial coincides with the probability of succeeding at most $m$ times out of $k$ trials, then the derivative of the latter probability with respect to $p$ must be larger than 1. This is formalized as follows. (Recall that $X(k,p)$ denotes a binomial random variable with $k\geq1$ trials and probability of success $p$ in each trial.)
\begin{prop} \label{pro-binom-deriv>1}
For any $k\geq2$ and $0\leq m<k$, there is a unique value $r_m\in(0,1)$ of $p$ for which 
$\Pr\left(X(k,p)\leq m\right)=p.$
Moreover, at $p=r_m$,
\[\left|\frac{d\Pr\left(X(k,p)\leq m\right)}{dp}\right|_{p=r_m}>1.
\]

\end{prop}
\begin{proof}

Fix some 
$k\geq 2$, and let 
\begin{equation}\label{eq:wkmp}
f_m(p)\equiv P(X(k,p) \leq m)=\sum_{i=0}^{m}\binom{k}{i}p^i(1-p)^{k-i}.
\end{equation}
We have
\begin{align*}
    f_m(p) + f_{k-m-1}(1-p) &= P(X(k,p)\leq m) + P(X(k,1-p)\leq k-m-1)\\
    &= P(X(k,p)\leq m) + P(X(k,p)\geq k- (k-m-1))\\
    &= P(X(k,p)\leq m) + P(X(k,p)\geq m+1) = 1.
\end{align*}
Let $r_m$ denote the fixed point of the function $f_m(\cdot),$ i.e., $f_m(r_m) =r_m.$ The fact that $f_{k-m-1}(1-p) = 1 - f_m(p)$ implies that $r_{k-m-1} = 1-r_m$ and $f_{k-m-1}'(r_{k-m-1}) = f_m'(r_m).$ Without loss of generality, we can therefore focus on analyzing the cases of $m$ for which $m \leq \lfloor\frac{k-1}{2}\rfloor.$ To complete the proof, we need to show that $|f_m'(r_m)| >1$ for $m \in \{0,1,\dots, \lfloor\frac{k-1}{2}\rfloor\}.$ In what follows, we show this.

We compute as follows:
\begin{align}\label{eq:wkmprime}
    f_{m}'(p) &= -\binom{k}{m}(k-m)p^m(1-p)^{k-m-1} = -k\binom{k-1}{m}p^m(1-p)^{k-m-1}\\
    f_{m}''(p) &= \binom{k}{m}(k-m)p^{m-1}(1-p)^{k-m-2}((k-m-1)p-m(1-p)).\nonumber
\end{align}
From the above computations, it follows that the function $f_m(\cdot)$ is concave for values of $p \leq p^{\ast}$ and convex for values of $p \geq p^{\ast}$ where $p^{\ast} = \frac{m}{k-1}.$ This is because $f_m''(p^{\ast}) = 0.$ Either the concave part or the convex part of the function $f_m(\cdot)$ intersects the $45^{\circ}$ line. Suppose that the concave part of the function $f_m(\cdot)$ intersects the $45^{\circ}$ line from the origin, i.e., $r_m \leq p^{\ast}.$ Since $m \leq \frac{k-1}{2},$ it follows that $r_m \leq 0.5.$ By concavity, we have
\[
f_m'(r_m)<\frac{f_m(0)-f_m(r_m)}{0-r_m}=-\left(\frac{1-r_m}{r_m}\right)\leq -1.\]
Therefore, we are done in cases where $r_m \leq p^{\ast}.$

We now consider the cases where $r_m > p^*$. By definition,
\begin{align}
   r_m &= (1-r_m)^k + \binom{k}{1}r_m(1-r_m)^{k-1} + \dots + \binom{k}{m}r_m^m(1-r_m)^{k-m}\label{eq:fixedpointpm}.
\end{align}
For $j = 0,1,2,\dots,m,$ let $a_j$ denote the $j^{\text{th}}$ term of the sum on the RHS of Eq. \eqref{eq:fixedpointpm}. For $j = 1,2,\dots, m,$ we compute as follows:
\begin{align*}
    &\frac{a_j}{a_{j-1}} = \frac{\binom{k}{j}r_m^j(1-r_m)^{k-j}}{\binom{k}{j-1}r_m^{j-1}(1-r_m)^{k-j+1}}=\left(\frac{k-j+1}{j}\right)\left(\frac{r_m}{1-r_m}\right)\\
    &\frac{a_j}{a_{j-1}}\geq 1 \iff (k-j+1)r_m \geq j(1-r_m) \iff r_m \geq \frac{j}{k+1}.
\end{align*}
Since $\frac{m}{k-1} > \frac{j}{k+1}$ for $j =1,2,\dots, m,$ we have $r \geq \frac{j}{k+1}$ and thus $a_j \geq a_{j-1}$. This implies that $a_j \leq a_m$ for $j=1,2\dots,m-1.$ By Eq.\eqref{eq:fixedpointpm}, we can thus infer the following:
\begin{equation}\label{eq:fixedpointinequality}
    r_m \leq (m+1)\binom{k}{m}r_m^m(1-r_m)^{k-m}.
\end{equation}
By Eqs. \eqref{eq:wkmprime} and \eqref{eq:fixedpointinequality}, we have
\[
|f_m'(r_m)| = (k-m)\binom{k}{m}r_m^m(1-r_m)^{k-m-1} \geq \left(\frac{k-m}{m+1}\right)\frac{r_m}{1-r_m}.
\]
Using the above set of equations, we can write a sufficient condition for $|f_m'(r_m)| > 1$  as follows:
\begin{equation}\label{eq:suffcondmodgeq1}
\left(\frac{k-m}{m+1}\right)\frac{r_m}{1-r_m} > 1 \iff r_m > \frac{m+1}{k+1}.
\end{equation}
We will now establish that $|f_m'\left(\frac{m+1}{k+1} \right)| > 1.$ By Eq. \eqref{eq:wkmprime}, we have
\begin{align*}
    |f_m'(p)| &= k\binom{k-1}{m}p^m(1-p)^{k-m-1} = k\cdot \Pr (X(k-1,p) = m).
\end{align*}
where $X(k-1,p)$ is a binomial distribution with parameters $k-1$ and $p.$ It is well known that the binomial distribution's mode with parameters $k-1$ and $p$ is attained at $\lfloor kp \rfloor.$ For $p = \frac{m+1}{k+1},$ we have 
\[
m \leq  k\cdot \left(\frac{m+1}{k+1}\right)  < m+1.
\]
The above inequalities imply that 
$\left \lfloor k\cdot \left(\frac{m+1}{k+1}\right) \right \rfloor = m.$ The binomial distribution $X_{k-1}(\cdot)$ has $k$ possible values and thus the probability of the occurrence of the mode has to be greater than $\frac{1}{k},$ i.e.,
\[
\Pr\left(X\left(k-1,\frac{m+1}{k+1}\right)= m \right) > \frac{1}{k} \implies \left|f_m'\left(\frac{m+1}{k+1} \right)\right| > 1.
\]
We need to consider the following two possible cases:

\textbf{Case 1:} $\frac{m}{k-1} < r_m \leq \frac{m+1}{k+1}$. 
For $1\leq m \leq \frac{k-1}{2},$ we know that $|f_m'(\cdot)|$ attains its maximum value at $\frac{m}{k-1}$ and that it is strictly decreasing for $p > \frac{m}{k-1}.$ For $m=0,$ $|f_m'(\cdot)|$ is strictly decreasing for $p>0.$ Thus, we have
\[
|f_m'(r_m)| > \left|f_m'\left(\frac{m+1}{k+1} \right)\right| > 1.
\]

\textbf{Case 2:} $r_m > \frac{m+1}{k+1}.$ Here, we are done by the sufficient condition of Eq. \eqref{eq:suffcondmodgeq1}.

\end{proof}
Recall that Proposition \ref{pro-binomial} states that there are no $p\neq q\in(0,1)$, such that $\Pr\left(X(k,p)\leq m\right)=q\in\left(0,1\right)$ and $\Pr\left(X(k,q)\leq m\right)=p.$ By contrast, Theorem \ref{thm:stableinterior_ASD} implies that this result does not hold if one replaces the binomial distribution with a mixture of binomial distributions. Formally:
\begin{cor}
For any $k\geq2$ and any $\alpha<\frac{1}{k}$, there exist $0\leq m<k$, $0\leq m'<k'$, and $p\neq q \in (0,1)$, such that 
$\alpha\cdot\Pr\left(X(k,p)\leq m\right)+(1-\alpha)\cdot \Pr\left(X(k',p)\leq m'\right)=q$\\
\hspace*{3.85cm}and  $\alpha\cdot\Pr\left(X(k,q)\leq m\right)+(1-\alpha)\cdot \Pr\left(X(k',q)\leq m'\right)=p.$
\end{cor}

We conclude this appendix with a proof of Proposition  \ref{pro-binomial}:  \begin{proof}[Proof of Proposition \ref{pro-binomial}]
Fix $k\geq 2$ and $m<k$, and for the same reason as in the proof of Proposition \eqref{pro-binom-deriv>1}, we can assume without loss of generality that $m\leq \lfloor\frac{k-1}{2}\rfloor.$  Let $f(p)\equiv\Pr\left(X(k,p)\leq m\right)$, and $r$ be the value of $p$ for which $f(r)=r$. Let $F(p)\equiv f(f(p))$. It is easy to see that if $f(p_1)=p_2$ and $f(p_2)=p_1$, then both $p_1$ and $p_2$ are fixed points of $F$. Moreover, since $f$ is decreasing, if $p_1<r$, then $p_2>r$. Therefore, it is enough to show that $F$ has no fixed points in $(0,r)$.

We claim first that if $p^*\equiv\frac{m}{k-1}\geq r$, then there is no fixed point $p_1<r$ for which $p_2\equiv f(p_1)\leq p^*$. In other words, $F$ has no fixed point in the range $[f^{-1}(p^*),r)$. To see this, assume by contradiction that there are such  fixed points  $p_1$ and $p_2$. Recall that the function $f$ is concave on $(0,p^*)$, and since $r<0.5$ and $|f'(r)|>1$ (Proposition \eqref{pro-binom-deriv>1}), the absolute value of the slope of the line connecting the points $(0,1)$ and $(p_2,p_1)$ is larger than one. Because of the concavity of $f$ on this range, the graph of $f$ is above this line, and the absolute value of the slope of the line connecting $(p_1,p_2)$ with $(p_2,p_1)$ (which is obviously equal to one) must be even larger of that of the line 
connecting the points $(0,1)$ and $(p_2,p_1)$, i.e., larger than 
$1=\left|\frac{p_1-p_2}{p_2-p_1}\right|$, a contradiction.


Now, the first and second derivatives of $F$ are
\begin{align*}
    F'(p) &=f'(f(p))\cdot f'(p)\\
    F''(p)&=f''(f(p))\cdot f'(p)^2+f'(f(p))\cdot f''(p).
\end{align*}
To complete the case of $p^*\geq r$ it is left to show that there are no fixed points of $F$ on (0,$f^{-1}(p^*)$). Indeed, for all $p\in (0,f^{-1}(p^*))$, $f''(f(p))>0$, 
which implies that $F''(p)>0$ (as $f''(p)<0$). Thus, on $(0,f^{-1}(p^*))$, $F(x)$ is convex and since $F(0)=0$, in this range there could be at most one fixed point of $F(p)$. Now, as was shown before, there are no fixed points of $F$ in  $[f^{-1}(p^*),r)$, and by Proposition \eqref{pro-binom-deriv>1}, $F'(r)>1$. This implies that the graph of $F(p)$ must be below the diagonal for any $p\in [f^{-1}(p^*),r)$. This could not happen if $F(p)$ were above the diagonal at some $p\in (0,f^{-1}(p^*))$.
We conclude that there is no such fixed point, and this completes the proof for the case of $p^*\geq r$.

Now assume $p^*<r.$ In this case, we will show that $F$ is convex on $(0,r)$. On the range $(0,p^*)$, $f''(f(p))>0$ and $f''(p)<0$, implying that $F''(p)>0$ and therefore $F$ are convex. On the range $(p^*,r)$, the condition for convexity is 
$$f''(f(p))\cdot f'(p)^2+f'(f(p))\cdot f''(p)>0, \textrm{ or}$$
$$\frac {f''(f(p))}{-f'(f(p))}>\frac {f''(p)}{f'(p)^2}.$$
Since $|f'(r)|>1$ (Proposition \ref{pro-binom-deriv>1}), $|f'(p)|>1$ for any $p\in(p^*,r)$, and therefore it is enough to show that for $p\in(p^*,r)$, 
$$\frac {f''(f(p))}{-f'(f(p))}>\frac {f''(p)}{-f'(p)}.$$
Let $G(p)\equiv \frac {f''(p)}{-f'(p)}$, and we will show that for any $p_1\in(p^*,r)$, $G(p)$ increases on $(p_1,f(p_1))$. Indeed, setting the explicit expressions of Eq. \eqref{eq:wkmprime} in $G$, we have  
$$G(p)=\frac{(k-1)p-m}{p(1-p)}=\frac{k-1-m}{1-p}-\frac{m}{p},$$
and since $m<k-1$, $G$ is increasing.

To summarize: since $F(0)=0$ and $F(r)=r$, when $p^*<r$, the convexity of $F$ on $(0,r)$ that we just showed implies that $F$ does not intersect the diagonal on $(0,r)$, and this completes the proof. 

\end{proof}
\subsection{Standard Definitions of Dynamic Stability}\label{sub-standard-definitions}
For completeness, we present in this appendix  the standard definitions of dynamic stability that is used in the paper (see, e.g., \citealp[Chapter 5]{weibull1997evolutionary}).

A state is said to be 
stationary if
it is a rest point of the dynamics.
\begin{defn}
\label{def:equilbirium}State $\mathbf{p}^{*}\in\left[0,1\right]^{2}$
is a \emph{stationary state}
if $w_{i}\left(\mathbf{p}^{*}\right)=p_{i}^{*}$
for each $i\in\left\{ 1,2\right\}$.
\end{defn}
Let $\mathcal{E}\left(w\right)$ denote the set of stationary states of $w,$ i.e., $\mathcal{E}\left(w\right) = \{\mathbf{p}^{*} | w_{i}\left(\mathbf{p}^{*}\right)=p_{i}^{*}\}.$ Under monotone dynamics,  an interior (mixed) state $\mathbf{p}^*\in (0,1)^2$ is a stationary state iff it is a Nash equilibrium (\citealp[Prop. 4.7]{weibull1997evolutionary}). By contrast, under nonmonotone dynamics (such as the sampling dynamics analyzed below) the two notions differ.

A state is Lyapunov stable if a population beginning near it remains close, and it is asymptotically stable if, in addition, it eventually converges to it. A state is
unstable if it is not Lyapunov stable. It is well known (see, e.g., \citealp[Section 6.4]{weibull1997evolutionary})
that every Lyapunov stable state must be a stationary state. Formally:

\begin{defn}
\label{def:lyaponouv-stability} A stationary state $\mathbf{p}^{*}\in\left[0,1\right]^{2}$
is \emph{Lyapunov stable} if for every neighborhood $U$ of $\mathbf{p}^{*}$
there is a neighborhood $V\subseteq U$ of $\mathbf{p}^{*}$ such that if the
initial state $p\left(0\right)\in V$, then $\mathbf{p}\left(t\right)\in U$
for all $t>0$. A state is \emph{unstable} if it is not
Lyapunov stable.
\end{defn}

\begin{defn}
A stationary state $\mathbf{p}^{*}\in\left[0,1\right]^{2}$
is \emph{asymptotically stable} if it is Lyapunov stable and there is some neighborhood $U$
of $\mathbf{p}^{*}$ such that all trajectories initially in $U$ converge
to $\mathbf{p}^{*},$ i.e., $\mathbf{p}\left(0\right)\in U$ implies $\lim_{t\rightarrow\infty}\mathbf{p}\left(t\right)=\mathbf{p}^{*}$.
\end{defn}

\subsection{Proof of Proposition \ref{pro-convergence-to-stationary} (Convergence to Stationary States) 
}\label{proof-convergence-to-stationary}
We say that state $\textbf{p}$ is \emph{above} (resp., \emph{below}) curve $w(p_1)$ if $p_2>w(p_1)$ (resp., ($p_2<w(p_1)$). Similarly, we say that state $\textbf{p}$ is to the \emph{right} (resp., \emph{left})  of the curve $w(p_1)$ if $w^{-1}(p_2)<p_1$ (resp., $w^{-1}(p_2)<p_1$). We say that the state $\textbf{p}$ is on the curve $w(p_1)$ if $p_2=w(p_1).$

The fact that the two curves of $w(p_1)$ and $w^{-1}(p_1)$ are strictly decreasing implies that a state is above one of these curves iff it is to the right of that curve. The states on the curve $w(p_1)$ (resp., $w^{-1}(p_1)$) are characterized by having $\dot{p_1}=0$ (resp., $\dot{p_2}=0$). Observe that in any state $\textbf{p}$ above and to the right (resp., below and to the left) of the curve $w(p_1)$, the share of hawks in population 2 is higher relative to the corresponding state on the curve with the same share of hawks in population 1, which implies that $\dot{p_1}<0$ (resp., $\dot{p_1}>0$). Similarly, in any state $\textbf{p}$ above (resp., below) the curve $w^{-1}(p_1)$, the  share of hawks in population 1 is higher relative to the corresponding state on the curve with same share of hawks in population 2, which implies that $\dot{p_1}<0$ (resp., $\dot{p_1}>0$).

Any state $\textbf{p}\in[0,1]$ can be classified in one of 3*3 classes, depending on its relative location with respect to the two curves, i.e., whether $\textbf{p}$ is below (i.e., to the left of), above (i.e., to the right of) or on the curve $w(p_1)$, and similarly whether $\textbf{p}$ is below, above, or on the curve $w^{-1}(p_1)$. 
The above argument implies that this classification determines the signs of $\dot {p_1}$ and $\dot {p_2}$. If state $\textbf{p}$ is on (resp., above, below) the curve $w(p_1)$, then $\dot {p_2}$ is zero (resp., negative, positive). Similarly, if state $\textbf{p}$ is on (resp., above, below) the curve $w^{-1}(p_1)$, then $\dot {p_1}$ is zero (resp., negative, positive).

In particular, any state  $\textbf{p}$ that is above (resp., below) both curves must satisfy that $\dot {p_1},\dot {p_2}$ are both negative (resp., positive). This implies that any trajectory that begins  above (resp., below) both curves (namely, $w(p_1)$ and $w^{-1}(p_1)$) must always move downward and to the left. This (together with the fact that all stationary states are in the intersection of the two curves) implies  that the trajectory must cross one of the curves. If this crossing point is an intersection of both curves, then it is a stationary state.

Next observe that any trajectory that starts on one of the curves and strictly above/below the remaining curve must move to a state that is strictly below one of the curves and strictly above the remaining curve. This is so because on the crossing point one of the $\dot {p_i}$ is zero and the remaining derivative $\dot{p_j}$ is negative (resp., positive), which implies that the trajectory must move below and to the left (resp., above and to the right) of the curve that was crossed.

Thus, we have established that any trajectory (which has not converged to a stationary state) must reach a state that is strictly below one of the curves and strictly above the remaining curve. If such a state is strictly above the curve $w(p_1)$, then by the classification mentioned above the trajectory must move both downward and to the right; i.e., $\dot{p_1},\dot{p_2}<0$ (resp., both upward and to the left; i.e., $\dot{p_1},\dot{p_2}>0$). This implies that the trajectory must cross one of the curves. The crossing point cannot be only on the curve of $w(p_1)$ (resp., $w^{-1}(p_1)$), because at such a point the trajectory moves horizontally to the left (vertically downward), which implies that it must cross the lower (resp., higher)  curve $w(p_1)$ (resp.,$w^{-1}(p_1)$) from the left side (resp., from above) and we get a contradiction. The argument in case the state is strictly below the curve $w(p_1)$ is analogous.
 
\subsection{Proof of Proposition \ref{pro-convergence-to-pure-iff-asymp} (Convergence to Pure Equilibria)}\label{proof-converge-pure-iss-asymp}
 If $\textbf{p}$ is below  (resp., above) both curves, then by the classification presented above it must that $\dot{p_1}>0$ (resp., $\dot{p_2}<2$), which implies (by a similar argument to the proof of Proposition \ref{pro-convergence-to-stationary}) that convergence to (1,0) is possible only if the trajectory passes through a state that is strictly between the two curves, and that the closest intersection point of the two curves to the left of this state is (1,0). By the classification presented above it must be that the curve of $w^{-1}(p_1)$ is strictly below the curve of $w(p_1)$ in a right neighborhood of (1,0), which implies that (1,0) is asymptotically stable because any sufficiently close initial state would converge to (1,0)
\subsection {Proof of Lemma \ref{lemma-sample} (Maximal Sample Sizes)}\label{lem-proof}

\begin{enumerate}
\item (I) The sum of payoffs of action $h$ (resp., $d$) against
a sample with a single $d$ is $1+g$ (resp., $1+\left(k-1\right)(1-l)$).
The mean payoff of $h$ is strictly higher than 
the mean payoff of $d$ iff
$k<\frac{1+g-l}{1-l}.$
(II) The sum of payoffs of action $d$ (resp., $h$) against
a sample with a single $h$ is $k-1+1-l$ (resp., $\left(k-1\right)\left(1+g\right)$).
The mean payoff of $h$ is higher than 
the mean payoff of $d$ iff $k<\frac{1+g-l}{g}$. 
\item (I) The sum of payoffs of an $h$-sample with a single $d$ is equal to $1+g$. The sum of payoffs of a $d$-sample with no $d$ is equal to $k\cdot l$. The former sum is
higher than the latter iff $k<\frac{1+g}{1-l}.$
(II) The sum of payoffs of a $d$-sample with no $h_{j}$-s is
equal to $k$. The sum of payoffs of an $h$-sample with a single
$h_{j}$ is equal to $(k-1)\left(1+g\right)$. The former sum
is higher than the latter iff $k<\frac{1+g}{g}$.\qedhere
\end{enumerate}

\subsection {Proof of Proposition \ref{prop:pure-equilbiria-satbility} (Stability of Pure Equilibria)}\label{subsec:Proof-of-Theorem-2}
We are interested in deriving conditions for the stability of the pure stationary states. 
In what follows,
we compute the Jacobian of the sampling dynamics for the state 
in which all agents of population $i$ play $d$ and all agents of population $j$ play $h$. For this, we consider a slightly perturbed state with a ``very small'' $\epsilon_i$ share of hawks in population $i$ and a ``very small'' $\epsilon_j$ share of doves in population $j.$ By ``very small,'' we mean that higher-order terms of $\epsilon_i$ and $\epsilon_j$ are neglected.

 Consider a new agent of population
$i$ with a sample size of $k_{i}.$ Action $h$ has a higher mean
payoff against a sample size of $k_{i}$ iff (neglecting rare events
of having multiple $d$-s in the sample): (1) the sample includes the single action $d$ of an opponent, and (2) $k_{i}\leq m_{h}$
 (due to Lemma \ref{lemma-sample}). The probability of (1) is $k_{i}\cdot\epsilon_{j}+o(\epsilon_{j})$,
where $o(\epsilon_{j})$ denotes terms that are sublinear in $\epsilon_{j}$,
and, thus, it will not affect the Jacobian as $\epsilon_{j}\rightarrow0$.
This implies that the probability that a new agent of population $i$
(with a random sample size distributed according to $\theta$)
has a higher mean payoff for action $h$ against her sample is
$w_{\theta}(1-\epsilon_{j})=\epsilon_{j}\cdot\sum_{k_{i}=1}^{m_{h}}\theta(k_i)\cdot k_{i}+o(\epsilon_{j})$.
An analogous argument implies that the probability that a new agent
of population $j$ has a higher mean payoff for action $d_{j}$ against
her sample is $w_{\theta_{j}}(\epsilon_{i})=\epsilon_{i}\cdot\sum_{k_{j}=1}^{m_{d}}\theta_{j}(k_j)\cdot k_{j}+o(\epsilon_{i})$.
Therefore, the sampling dynamics at $(\epsilon_{i},1-\epsilon_{j})$
can be written as follows (ignoring the higher-order terms
of $\epsilon_{i}$ and $\epsilon_{j}$):
\begin{equation}
\dot{\epsilon}_{i}=\epsilon_{j}\cdot\sum_{k_{i}=1}^{m_{h}}\theta(k_i)\cdot k_{i}-\epsilon_{i},\,\,\,\,\,\,\,\,\,\,\,\,\dot{\epsilon}_{j}=\epsilon_{i}\cdot\sum_{k_{j}=1}^{m_{d}}\theta_{j}(k_j)\cdot k_{j}-\epsilon_{j}.\label{eq:a_b_theta}
\end{equation}
Define: $a_{\theta}=\sum_{k_{i}=1}^{m_{h}}\theta(k_i)\cdot k_{i}$
and $b_{\theta}=\sum_{k_{j}=1}^{m_{d}}\theta_{j}(k_j)\cdot k_{j}$.
The Jacobian of the above system of Equations \eqref{eq:a_b_theta}
is then given by $J=\left(\begin{array}{cc}
-1 & a_{\theta}\\
b_{\theta} & -1
\end{array}\right).$ The eigenvalues of $J$ are $-1-\sqrt{a_{\theta}b_{\theta}}$
and $-1+\sqrt{a_{\theta}b_{\theta}}.$ Observe that:
(1) if $a_{\theta}b_{\theta}<1$ then both eigenvalues
are negative, which implies that the state in which
all agents in population $i$ (resp., $j$) are doves (resp., hawks) is asymptotically
stable, and (2) if $a_{\theta}b_{\theta}>1$ then
one of the eigenvalues is positive, which implies that this state is unstable (see, e.g., \citealp[Theorems 1--2 in Section 2.9]{perko2013differential}).

\subsection{Proof of Corollary \ref{cor:g=l}\label{proof-cor-g=l} (Standard Case of $g=l$)}
Substituting $g=l$ in  the bounds presented in Definition \ref{def:maximal-samples} yields the following bounds: 
 $\,m_{h}^{A}=\left\lfloor \frac{1}{1-g}\right\rfloor ,\,m_{d}^{A}=\left\lfloor \frac{1}{g}\right\rfloor,
 m_{h}^{P}=\left\lfloor \frac{1+g}{1-g}\right\rfloor ,\,m_{d}^{P}=\left\lfloor \frac{1+g}{g}\right\rfloor.$
 Observe that $\min(m_{h}^{A},m_{d}^{A})=\min( \frac{1}{1-g}, \frac{1}{g})=1$. This implies that  
 $\mathbb{E}_{\leq m_{h}^A}\left(\theta\right)\cdot \mathbb{E}_{\leq m_{d}^A}\left(\theta\right)=\theta(1)\cdot\mathbb{E}_{\leq \max\left(\frac{1}{g},\frac{1}{1-g}\right)}\left(\theta\right),$ which proves part 1.
 Next observe that if $g<\frac{1}{3}$, then  $\min(m_{h}^{P},m_{d}^{P})=\min(\frac{1+g}{1-g},\frac{1+g}{g})=1$ and $\max(m_{h}^{P},m_{d}^{P})=\max(\frac{1+g}{1-g},\frac{1+g}{g})=\frac{1+g}{g}$, which proves part (2-a). Finally, observe that if $g\geq\frac{1}{3}$, then 
 $\min(m_{h}^{P},m_{d}^{P})=\min(\frac{1+g}{1-g},\frac{1+g}{g})=2$ and $\max(m_{h}^{P},m_{d}^{P})=\max(\frac{1+g}{1-g},\frac{1+g}{g})=\max(3,\frac{1+g}{1-g})$, which proves part (2-b).

\subsection{Proof of Theorem \ref{thm:stableinterior_ASD} (Stable Symmetric State, ASD)}\label{subsec:stableinterior_ASDproof}
The following notation will be helpful. For each $k>1$ and $q\in(0,1)$, let $\underline{p}_k^q$ (resp., $\overline{p}_k^q$) be the unique solution in $(0,1)$ to the equation $p=q\cdot(1-p)^k$ (resp., $p=1-q\cdot p^k$). We interpret $\underline{p}_k^q$ (resp., $\overline{p}_k^q$) as the symmetric stationary state when $g$ and $l$ are close to zero (resp., close to one) in a population in which a share $q$ of the population have samples of size $k$, while the remaining agents always play dove (resp., hawk).

Let $\overline{g}\in(0,1)$ (resp., $\overline{l}\in(0,1)$) be sufficiently small (resp., large) such that  $$(k-1)\overline{g}<1-l \textrm{ and } \frac{\overline{g}}{1+\overline{g}-l}<\underline{p}_k^q$$ $$(\textrm{resp.,} (k-1)(1-\overline{l})<g \textrm{ and }  \frac{g}{1+g-\overline{l}}>\overline{p}_k^q).$$  Fix any $0<g<\overline{g}$ (resp., any $\overline{l}<l<1$). The definition of $\underline{p}_k^q$ (resp.,  $\overline{p}_k^q$) implies that the symmetric stationary state $(p_\theta,p_\theta)$ in any population with $\theta(k)=q$ must satisfy $p_\theta>\underline{p}_k^q$ (resp., $p_\theta<\overline{p}_k^q$). Observe that the best reply to the true distribution of hawks in the population in state $(p_\theta,p_\theta)$
is to play $d$ (resp., $h$) because $\frac{g}{1+g-l}<p_\theta$ (resp., because $\frac{g}{1+g-l}>p_\theta$). The central limit theorem implies that new agents with sufficiently large samples almost always play $d$ (resp., $h$) in a neighborhood of $p_\theta.$ This implies that for each $\epsilon>0,$ there exists $\overline{k}$ such that $|w'_n(p_\theta)|<\epsilon$ for any $n\geq\overline{k}$. 
Let $\epsilon>0$ be sufficiently small such that $q\cdot k<1-\epsilon$.  This implies that:
$$|w'_\theta(p_\theta)|<(1-q)\cdot\epsilon+q\cdot|w'_k(p_\theta)|<\epsilon+q\cdot|((1-p)^k)'|=\epsilon+q\cdot|k\cdot(1-p)^{k-1}|<\epsilon+q\cdot k<1.$$
$$\left(|w'_\theta(p)|=<(1-q)\cdot\epsilon+q\cdot|w'_k(p_\theta)|<\epsilon+q\cdot|(1-p^k)'|=\epsilon+q\cdot|k\cdot{p}^{k-1}|<\epsilon+q\cdot k<1\right.)$$ This implies that $(p_\theta,p_\theta)$ is asymptotically stable due to  Proposition \ref{prop:pure-equilbiria-satbility}.   
\subsection{Proof of Theorem \ref{thm:stableinterior_PSD} (Stable Symmetric State, PSD)}\label{subsec:stableinterior_PSDproof}
Recall that the payoff-sampling dynamics in state $(p_1,p_2)$ are given by
$\dot{p}_1 = \delta(w_\theta(p_2) - p_1$ and $\dot{p}_2  = \delta(w_\theta(p_1) - p_2),$
and that a symmetric state $(p^{(\theta)},p^{(\theta)})$ is  stationary iff $p^{(\theta)}$ $w_\theta (p^{(\theta)})=p^{(\theta)}$. 


We now establish some properties of the payoff-sampling dynamics and the symmetric rest points for symmetric distributions of types $\theta \equiv k$. 
If $l,1-g \in \left(0, \frac{1}{\max(\textrm{supp}(\theta))}\right),$ action $h_i$ has a higher mean payoff iff the number of $d_j$-s in the $h_i$-sample is strictly greater than half the number of $d_j$-s in the $d_i$-sample. To express $w_k(p)$ concisely in this case, we define $X(k,p)$ and $Y(k,p)$ to be independent and identically distributed binomial random variables with parameters $k$ and $p.$ We can then write $w_k(p)$ as follows:
\begin{equation}
w_k(p) = P\left(k-X(k,p) > \frac{1}{2}(k-Y(k,p))\right)= P(2X(k,p)-Y(k,p) < k).
\end{equation}
Observe that $w_k(p)$ is a polynomial in $p$ of degree at most $2\cdot k.$ We have verified the following facts about these polynomials for $k<20$ (for an illustration see Figure \ref{fig:proofbypicture}; the Mathematica code is given in the online supplementary material, and the explicit values of the rest points and the derivatives are presented in Table \ref{table:fixedpoints}):
\begin{itemize}
    \item For $k \in \left\{1,2,\dots,18,19\right\},$ $w_k(p)$ is decreasing in $p.$
    \item For $k \in \left\{1,2,\dots,18,19\right\},$ $w_k(p)$ has a unique fixed point $p^{(k)}$.\\
    Moreover, $0.5 < p^{(k)} < 0.68$ for any $k \in \left\{1,2,3,4,5\right\}.$ 
    \item $|w_1'(p)| \equiv 1,$ and $|w_k'(p)| < 1$ for any $k \in \left\{2,3,4,5\right\}$ and $0.5<p<0.68.$ 
\end{itemize}

Recall that $w_\theta(p)$ is a convex combination of the $w_k(p)$ for the $k$-s in its support 
(i.e., $w_\theta(p)=\sum_{k}\theta\left(k\right)\cdot w_{k}\left(p\right)$).  From the above facts, it follows that:
\begin{enumerate}
    \item  For $\theta \equiv k < 20,$ the function $w_k(p)$ has a unique fixed point $p^{(k)}$ such that $|w'_k(p^{(k)})| < 1,$ which implies that $(p^{(k)}, p^{(k)})$ is asymptotically stable.
    \item For $\max(\textrm{supp}(\theta)) \leq 5,$ the function $w_{\theta}(p)$ has a unique fixed point $p^{(\theta)}$ such that $p^{(\theta)} \in (0.5,0.68)$ and $|w'_\theta(p^{(\theta)})| < 1$ if $\theta(1) \neq 1,$ 
    which implies that $(p^{(\theta)}, p^{(\theta)})$ is asymptotically stable.
\end{enumerate}

\begin{table}[h]
\caption{Fixed Points of the Function $w_k(p)$ in the Proof of Theorem \ref{thm:stableinterior_PSD}} 
\centering

\begin{tabular}{|c|c|c|c|c|c|c|c|c|c|c|}
\hline 
$k$ & 1 & 2 & 3 & 4 & 5 & 6 & 7 & 8 & 9 & 10\tabularnewline
\hline 
$p^{\left(k\right)}$ & 0.500 & 0.579 & 0.620 & 0.649 & 0.672 & 0.690 & 0.706 & 0.720 & 0.731 & 0.741\tabularnewline
\hline 
$\ensuremath{|w'_{k}(p^{\left(k\right)})|}$ & 1 & 0.690 & 0.618 & 0.645 & 0.690 & 0.730 & 0.763 & 0.793 & 0.818 & 0.840 \tabularnewline
\hline 
\end{tabular}

\medskip{}

\begin{tabular}{|c|c|c|c|c|c|c|c|c|c|c|}
\hline 
$k$ & 11 & 12 & 13 & 14 & 15 & 16 & 17 & 18 & 19 & 20\tabularnewline
\hline 
$p^{\left(k\right)}$ & 0.750 & 0.758 & 0.765 & 0.773 & 0.778 & 0.784 & 0.789 & 0.794 & 0.799 & 0.803\tabularnewline
\hline 
$\ensuremath{|w'_{k}(p^{\left(k\right)})|}$ & 0.861 & 0.88 & 0.899 & 0.916 & 0.932 & 0.948 & 0.963 & 0.978 & 0.991 & 1.001 \tabularnewline
\hline 
\end{tabular}

 \label{table:fixedpoints} 
\end{table}

\begin{table}[h]
\caption{Values of $|w_k'(p^{(j)})|$ for $k,j \in \{1,2,3,4,5\}.$} 
\centering
\begin{tabular}{ |p{0.7cm}|p{1.5cm}|p{1.5cm}|p{1.5cm}|p{1.5cm}|p{1.5cm}|  }
 \hline
 $_{k} \backslash ^{p^{(j)}}$& $p^{(1)}$ &$p^{(2)}$&$p^{(3)}$ &$p^{(4)}$&$p^{(5)}$\\
 \hline
 1   & 1   &1 & 1    &1 & 1\\
 \hline
 2   & 0.5    &0.690 & 0.812    &0.905& 0.981\\
 \hline
 3   & 0.562    &0.560 & 0.618    &0.687 & 0.759\\
 \hline
 4   & 0.625    &0.616 & 0.623    &0.645 & 0.679\\
 \hline
 5   & 0.605    &0.642 & 0.659    &0.673 & 0.690\\
 \hline
\end{tabular}
 \label{table:absvalue_fixpoints} 
\end{table}
\spacing{1.213}
\bibliographystyle{chicago}
\bibliography{mybibdata}
\end{document}